\documentclass[11pt]{article}
\usepackage{fullpage}
\usepackage{algorithm}
\usepackage[noend]{algpseudocode}
\usepackage{amsmath,amsthm,amssymb}
\usepackage{hyperref}
\usepackage{cleveref}

\newtheorem{lemma}{Lemma}
\newtheorem{corollary}{Corollary}
\newtheorem{definition}{Definition}
\usepackage{authblk}

\usepackage{xspace}
\usepackage{mathtools}
\usepackage{csquotes}
\usepackage{xcolor}
\usepackage{longtable}
\usepackage{skull}
\usepackage{subcaption}

\newcommand{\set}[1]{\protect\ensuremath{\left\{ #1\right\}}\xspace}

\DeclareMathOperator{\dd}{{\textup{d}}}
\DeclareMathOperator{\diam}{{\textup{diam}}}
\DeclareMathOperator{\radius}{{\textup{rad}}}
\DeclareMathOperator{\ecc}{{\textup{ecc}}}
\DeclareMathOperator{\OO}{{\cal O}}
\DeclareMathOperator{\mw}{mw}
\DeclareMathOperator{\sw}{sw}

\def\computeAccVal{\texttt{computeAccVal}}
\def\HRS{\texttt{hierarchicalDominatingSet}}
\def\mates{\texttt{mates}}

\newtheorem{observation}{Observation}

\crefname{algocf}{alg.}{algs.}
\Crefname{algocf}{Algorithm}{Algorithms}

\title{Hyperbolicity Computation through Dominating Sets\thanks{This work has been supported by the French government, through the UCA$^\textsc{jedi}$ Investments in the Future project managed by the National Research Agency (ANR) with the reference number ANR-15-IDEX-01, the ANR project Multimod with the reference number ANR-17-CE22-0016 and the ANR project Distancia with reference number ANR-17-CE40-0015.}} %TODO Please add

\author[1]{David Coudert}
\author[2]{Andr\'e Nusser}
\author[3]{Laurent Viennot}

\affil[1]{Universit\'e C\^ote d'Azur, Inria, CNRS, I3S, France}
\affil[2]{Max Planck Institute for Informatics and Graduate School of Computer Science, Saarland Informatics Campus, Saarbrücken, Germany}
\affil[3]{Inria, Paris University, CNRS, Irif, France}

\date{\vspace{-5ex}}

\begin{document}

\maketitle

\maketitle

\begin{abstract}
Hyperbolicity is a graph parameter related to how much a graph resembles a tree with respect to distances. Its computation is challenging as the main approaches consist in scanning all quadruples of the graph or using fast matrix multiplication as building block, both are not practical for large graphs. In this paper, we propose and evaluate an approach that uses a hierarchy of distance-$k$ dominating sets to reduce the search space. This technique, compared to the previous best practical algorithms, enables us to compute the hyperbolicity of graphs with unprecedented size (up to a million nodes)  and speeds up the computation of previously attainable graphs by up to 3 orders of magnitude while reducing the memory consumption by up to more than a factor of 23.

\end{abstract}

%%%%%%%%%%%%%%%%%%%%%%%%%%%%%%%%%%%%%%%%%%%%%%%%%%%%%%
%%%%%%%%%%%%%%%%%%%%%%%%%%%%%%%%%%%%%%%%%%%%%%%%%%%%%%
%%%%%%%%%%%%%%%%%%%%%%%%%%%%%%%%%%%%%%%%%%%%%%%%%%%%%%

\section{Introduction}
\label{sec:intro}

This paper aims at computing the hyperbolicity
of graphs with millions of nodes. 
Hyperbolicity is a graph parameter used to classify complex networks~\cite{Abu2015,AlrasheedD15,Kennedy2013}. It can be seen as a measure of how much a network is ``democratic''~\cite{ADM14,BCC15}. It has recently received growing attention as it
appears to capture important properties of several large practical graphs such as Internet~\cite{ShavittT04}, the Web~\cite{MunznerB95} and database relations~\cite{WalterR02}.
%Formal relationships between Gromov hyperbolicity and the existence of a core (a subset of vertices intersecting a constant fraction of all the shortest-paths) are investigated in~\cite{ChepoiDV17}. Reciprocally, the existence of a core is shown to be inherent to any hyperbolic network in~\cite{ChepoiDV17}.
Additionally, measuring hyperbolicity has applications in routing~\cite{Boguna2010,ChepoiDEHVX12,Krauthgamer2006}, %approximating other graph parameters~\cite{Chepoi2008,DasGuptaKMY18} 
and bioinformatics~\cite{Chakerian2012,Dress2012}.
See~\cite{Abu2015,Dragan2013} for recent surveys.

Hyperbolicity is usually defined through a 4-points condition that associate to each quadruple a $\delta$-value defined through distances between the four nodes (see Section~\ref{sec:def-notations} for the formal definition). The hyperbolicity of a graph is the maximum of the $\delta$-value over all quadruples in the graph. This definition trivially results in a $\Theta(n^4)$ algorithm for computing hyperbolicity, and the best known theoretical complexity is $O(n^{3.69})$, relying on an optimized (max,min)-matrix product~\cite{FournierIV15}.
However, the algorithms exhibiting the best performances in practice have time complexity in $\OO(n^4)$-time~\cite{BorassiCCM15,CohenCL15,coudert:hal-03201405}.

The quest for computing the hyperbolicity of large graphs has lead to gain orders of magnitude in the size of the considered graphs.
Indeed, the first attempts are based on brute force implementations of the trivial $\Theta(n^4)$ algorithm~\cite{Chakerian2012,Adcock2013}.
More precisely, \cite{Chakerian2012,distory} uses the revolving door algorithm to explore all quadruples, thus enabling to compute the hyperbolicity of graphs with few hundreds of nodes. Then, using massive parallelism (up to 1\,000 cores), \cite{Adcock2013} was able to consider graphs with up to 8\,000 nodes, but this approach is not scalable.
The first noticeable progress is due to~\cite{CohenCL15} which introduces pruning techniques to drastically reduce the number of quadruples to consider. With the addition of refined pruning techniques~\cite{BorassiCCM15,SotoGomez2011}, it enables to compute the hyperbolicity of graphs with up to 50\,000 nodes.
Furthermore, the running time of this algorithm using a single core is orders of magnitude smaller  than the running times reported in~\cite{Adcock2013}.
However, this algorithm reaches a memory bottleneck as it has space complexity in $\OO(n^2)$. 
To go beyond this bottleneck, \cite{coudert:hal-03201405} engineered an algorithm, along with suitable data structures, that consumes significantly less memory while offering good performances in practice. It enables to compute the hyperbolicity of graphs with more than 100\,000 nodes. Nonetheless, the memory usage of this algorithm is still high, which limits its scalability.
In this paper, we propose a new approach that uses a hierarchy of distance-$k$ dominating sets to
both reduce memory usage and further prune the search space, resulting in an algorithm enabling to compute for the first time the hyperbolicity of graphs with up to a million nodes.

\subsection{Our approach.}
In this paper, we propose to reduce the search space through the use of a hierarchy of distance-$k$ dominating sets. The approach applies to graphs having small distance-$k$ dominating sets for small $k$, which is often the case in practice. Recall that a distance-$k$ dominating set is a set $D$ such that any node of the graph is at distance at most $k$ from a node in $D$. 
The main idea is to explore quadruples within the dominating set. Only when the $\delta$-value of such a quadruple is large enough (depending on the highest $\delta$-value found so far and $k$), we explore recursively the quadruples dominated by it. More precisely, each node is associated to its closest dominator, that is a node of $D$ at distance at most $k$. Each quadruple of the graph is said to be dominated by the four associated dominators. For each dominator, we compute recursively a distance-$k'$ dominating set of the nodes associated to it with $k'<k$, and repeat recursively this process with smaller and smaller values of $k'$. The search explores recursively quadruples of dominators in distance-$k'$ dominating sets with smaller and smaller values of $k'$. At the deepest level, the search uses $k'=0$ and all quadruples with potential high $\delta$-value are explored. Although many quadruples are usually skipped, the quadruples with highest $\delta$-value are guaranteed to be visited and the computation is exact. Besides having small distance-$k$ dominating sets, the efficiency of the method also relies on the property that most of the quadruples have a relatively small $\delta$-value which seems to be the case in practice when the hyperbolicity of the graph is large enough compared to $k$. More precisely, the graph should have hyperbolicity greater than $2k$ and a $k$-dominating set of few thousand nodes at most. 

\subsection{Main contributions.}

Our main contribution is a new approach for scanning nodes of a graph through a hierarchy of dominating sets.
In particular we show its effectiveness for computing hyperbolicity through a set of technical lemmas allowing to relate the delta value of a quadruple to that of four nodes dominating them in a distance-$k$ dominating set. 
We also provide an implementation~\cite{gitlabv2} proving the efficiency of the method by computing the hyperbolicity of practical graphs with unprecedented size.

\subsection{Other related works.}
Graphs with hyperbolicity at most 1 can be recognized efficiently. Indeed,  0-hyperbolic graphs are block graphs (graphs whose biconnected components are cliques)~\cite{Bandelt1986,Howorka1979}, which can be recognized in time $\OO(n+m)$. Furthermore, using a 2-approximation algorithm, one can decide in time $\OO(n^{2.69})$ if a graph has hyperbolicity at most 1~\cite{FournierIV15}. Finally, deciding if a graph is  $\frac{1}{2}$-hyperbolic is equivalent to decide if it contains an induced 4-cycle~\cite{CoudertD14}, which can be done in time $\OO(n^{3.26})$~\cite{LeGall2012}.

Preprocessing methods for reducing the size of the input graph have been proposed.
In particular, the hyperbolicity of a graph is the maximum over the hyperbolicity of its biconnected components (see~\cite{CohenCL15} for a proof). Taking the maximum hyperbolicity over the atoms of a decomposition of a graph by clique-minimal separators~\cite{Tarjan1985,Berry2010b} results in an additive +1 approximation of hyperbolicity and~\cite{CohenCDL17} shows how to modify the atoms to get the exact hyperbolicity of the graph.
Moreover, \cite{SotoGomez2011} proved that the hyperbolicity of a graph is the maximum of the hyperbolicity of the graphs resulting from either a \emph{modular}~\cite{Gallai1967,Habib2010} or a \emph{split}~\cite{Cunningham1980,Cunningham1982} decomposition of the graph and~\cite{CoudertDP19} proposed algorithms with time complexity in $\OO(\mw(G)^{3}\cdot n+m)$ and $\OO(\sw(G)^{3}\cdot n+m)$, where $\mw(G)$ is the \emph{modular width} and $\sw(G)$ the \emph{split width}. Observe that these last decompositions can be computed in time~$\OO(n+m)$~\cite{CharbitMR12}. 

%Several heuristics have also been proposed, either based on the random sampling of quadruples~ \cite{Kennedy2013,KennedySN16}, or on successive graph traversals~\cite{CohenCL15}. 

\subsection{Organisation.}

\S~\ref{sec:def-notations} introduces the formal definitions of hyperbolicity, $\delta$-value and distance-$k$ domination. \S~\ref{sec:domhyp} provides the technical lemma relating the $\delta$-value of a quadruple to that of four nodes dominating it. \S~\ref{sec:extending_borassi} provides technical lemmas for extending~\cite{BorassiCCM15} as well as our new algorithm. Some experiments are presented in \S~\ref{sec:experimentation} before concluding with \S~\ref{sec:conclusion}.

%%%%%%%%%%%%%%%%%%%%%%%%%%%%%%%%%%%%%%%%%%%%%%%%%%%%%%
%%%%%%%%%%%%%%%%%%%%%%%%%%%%%%%%%%%%%%%%%%%%%%%%%%%%%%
%%%%%%%%%%%%%%%%%%%%%%%%%%%%%%%%%%%%%%%%%%%%%%%%%%%%%%

\section{Definitions and notations}
\label{sec:def-notations}

We use the graph terminology of~\cite{Bondy1976,Diestel1997}.
All graphs considered in this paper are finite, connected, unweighted and simple.
The graph $G=(V,E)$ has $n = |V|$ vertices and $m = |E|$ edges.
The open neighborhood $N_G(S)$ of a set $S \subseteq V$ consists of all vertices in $V \setminus S$ with at least one neighbor in $S$.
%The closed neighborhood of $S$ is the set $N_G[S] = S \cup N_G(S)$.

Given two vertices $u$ and $v$, a \emph{$uv$-path} of length $\ell \geq 0$ is a sequence of vertices $(u=v_0v_1\ldots v_{\ell}=v)$, such that $\{ v_i, v_{i+1} \}$ is an edge for every $i$.
In particular, a graph $G$ is \emph{connected} if there exists a $uv$-path for all pairs $u,v \in V$, and in such a case the \emph{distance} $\dd_G(u,v)$ is defined as the minimum length of a $uv$-path in $G$.
When $G$ is clear from the context, we write $\dd$ (resp. $N$) instead of $\dd_G$ (resp. $N_G$).
The \emph{eccentricity} $\ecc(u)$ of a vertex $u$ is the maximum distance between $u$ and any other vertex $v\in V$, i.e., $\ecc(u) = \max_{v\in V} \dd(u,v)$. The maximum eccentricity is the \emph{diameter} $\diam(G)$ and the minimum eccentricity is the \emph{radius} $\radius(G)$.

\textbf{Domination.}
Given an integral $k>0$, we say that a node $u$ \emph{$k$-dominates} a node $v$ when $\dd(u,v)\le k$.
We define a \emph{$k$-dominating set} of $G$ as a set $D\subseteq V$ of nodes such that any node $v\in V$ is \emph{$k$-dominated} by some node $u\in D$. Given a $k$-dominating set $D$, we associate to any node $v\in V$ such a $k$-dominating node $D(v)$ in $D$. $D(v)$ is called the \emph{associated dominator} of~$v$. 
We denote $D^{-1}(u)$ the set of vertices that are $k$-dominated by $u\in D$, i.e., $D^{-1}(u)=\set{v \in V : D(v) = u}$.
For each node $u\in D$, we define its \emph{domination radius} as $k_u=\max_{v\in D^{-1}(u)} \dd(u,v)$.
When a node $v$ is $k$-dominated by several nodes of $D$, the choice of its associated dominator $D(v)$ can be arbitrary. However, we will typically choose $D(v)$ as a closest node to $v$ in $D$ as a heuristic to obtain smaller values of $k_u$ with $u\in D$.

\textbf{Hyperbolicity.} This notion has been introduced to measure how the shortest-path metric space $(V, \dd_G)$ of a connected graph $G=(V, E)$ deviates from a tree metric when its vertices are mapped into the vertices of an edge-weighted tree. This additive stretch of the distances, denoted $\delta$, is called the \emph{hyperbolicity} of the graph and a graph is said to be  $\delta$-hyperbolic if it satisfies the 4-point condition below.

\begin{definition}[\textbf{$4$-points Condition,~\cite{Gromov1987}}]
\label{def:hyperbolic}
  Let $G$ be a connected graph.
  For every quadruple $u,v,x,y$ of vertices of $G$, we define $\delta(u,v,x,y)$ as half of the difference between the two largest sums among
  $S_1=\dd(u,v)+\dd(x,y)$, $S_2 = \dd(u,x)+\dd(v,y)$, and $S_3 = \dd(u,y)+\dd(v,x)$.

  The hyperbolicity of $G$, denoted by $\delta(G)$, is equal to $\max_{u,v,x,y\in V(G)} \delta(u,v,x,y)$.
Moreover, given a value $\overline{\delta}$, we say that $G$ is \emph{$\overline{\delta}$-hyperbolic} whenever $\delta(G) \leq \overline{\delta}$.
\end{definition}

Note that if $G$ is  a tree, or a clique, 
one can easily show $S1=S2$ for all quadruples, implying $\delta(G)=0$.
If $G$ is a cycle of order $n = 4p + \varepsilon$, with $p \geq 1$ and $0 \leq \varepsilon < 4$, then $\delta(G) = p - 1/2$ when $\varepsilon = 1$, and $\delta(G)=p$ otherwise. 
If $G$ is a $n \times m$ grids, with $2 \leq n \leq m$, then we have $\delta(G) = n - 1$.
Other definitions of hyperbolicity have been proposed~\cite{Bermudo2013,LaHarpe1990,Gromov1987} and differ only by a small constant factor.

%%%%%%%%%%%%%%%%%%%%%%%%%%%%%%%%%%%%%%%%%%%%%%%%%%%%%%
%%%%%%%%%%%%%%%%%%%%%%%%%%%%%%%%%%%%%%%%%%%%%%%%%%%%%%
%%%%%%%%%%%%%%%%%%%%%%%%%%%%%%%%%%%%%%%%%%%%%%%%%%%%%%

\section{Using dominating set for approximating hyperbolicity}
\label{sec:domhyp}

In this section, we show how to obtain an additive $+4k$ approximation of hyperbolicity from a $k$-dominating set. 

\begin{lemma}\label{lem:dominating}
Given a $k$-dominating set $D$ of $G$, 
a quadruple $u',v',x',y'\in V$ with respective associated dominators $u,v,x,y\in D$,
%% such that $\delta(u',v',x',y') \geq K_4$ where\\
we then have $\delta(u,v,x,y) - K_4 \le \delta(u',v',x',y')\le \delta(u,v,x,y) + K_4$ where $K_4 = k_{u} + k_{v} + k_{x} + k_{y}$.
\end{lemma}

\begin{proof}
We first prove the second inequality.
Assume without loss of generality that the first sum for quadruple $u',v',x',y'$ is $S_1'=\dd(u',v')+\dd(x',y')$, the second sum is $S_2'=\dd(u',x')+\dd(v',y')$ and the third sum is $S_3'=\dd(u',y')+\dd(v',x')$ where $S_1'\ge S_2'\ge S_3'$. We let $S_1,S_2,S_3$ denote the corresponding sums with $u,v,x,y$: $S_1=\dd(u,v)+\dd(x,y)$, $S_2=\dd(u,x)+\dd(v,y)$ and  $S_3=\dd(u,y)+\dd(v,x)$. 

When $\delta(u',v',x',y')\le K_4$, we obviously have $\delta(u',v',x',y')\le \delta(u,v,x,y) + K_4$ as $\delta(u,v,x,y)\ge 0$.

Now consider the case $\delta(u',v',x',y')\ge K_4$. We can then show $S_1\ge S_2$ and $S_1\ge S_3$ as follows.
Note that by triangle inequality, we have $\dd(u',v')\le \dd(u',u)+\dd(u,v)+\dd(v,v')\le \dd(u,v)+k_u+k_v$.
Using similar inequalities for other pairs, we can deduce $S_1-S_2\ge S_1' - K_4 - (S_2' + K_4)\ge S_1' - S_2' - 2K_4$.
By hyperbolicity definition, we have $S_1'-S_2'=2\delta(u',v',x',y')$, and we get $S_1-S_2\ge 0$ as we assume $\delta(u',v',x',y')\ge K_4$. Similarly, we have $S_1-S_3\ge S_1'-S_3'-2K_4\ge S_1'-S_2'-2K_4$ since we have $S_3'\le S_2'$.
We thus also have $S_1-S_3\ge 0$. 

We finally prove $\delta(u',v',x',y') \le \delta(u,v,x,y) + K_4$. Having $S_1\ge S_2$ and $S_1\ge S_3$, the greatest sums for quadruple $u,v,x,y$ are either $S_1$ and $S_2$ or $S_1$ and $S_3$, implying $2\delta(u,v,x,y)=\min\{S_1-S_2, S_1-S_3\}$. As both $S_1-S_2$ and $S_1-S_3$ are at least $S_1'-S_2'-2K_4$ from above inequalities, we obtain $2\delta(u,v,x,y) \ge S_1'-S_2'-2K_4 = 2\delta(u',v',x',y')-2K_4$ as desired.

\medskip

We now prove the first inequality. We assume without loss of generality that the first sum for quadruple $u,v,x,y$ is $S_1=d(u,v)+d(x,y)$, the second sum is $S_2=d(u,x)+d(v,y)$ and the third sum is $S_3=d(u,y)+d(v,x)$. When $\delta(u,v,x,y)\le K_4$, we obviously have $\delta(u,v,x,y) - K_4 \le \delta(u',v',x',y')$ as $\delta(u',v',x',y')\ge 0$. Otherwise, we can prove similarly as before that we have $S_1'-S_2'\ge S_1-S_2-2K_4\ge 0$ and $S_1'-S_3'\ge S_1-S_3-2K_4\ge S_1-S_2-2K_4\ge 0$. This ensures that $S_1'$ is the greatest sum for $u',v',x',y'$ in that case too. We thus have $2\delta(u',v',x',y')=\min\{S_1'-S_2', S_1'-S_3'\}\ge \min\{S_1-S_2-2K_4, S_1-S_3-2K_4\}=S_1-S_2-2K_4=2\delta(u,v,x,y) - 2K_4$.
\end{proof}

From Lemma~\ref{lem:dominating}, we can get an additive $+4k$ approximation of the hyperbolicity of the graph.
\begin{corollary}\label{cor:dominating}
Given a $k$-dominating set $D$ of $G$, let $\delta_L = \max_{u,v,x,y\in D} \delta(u,v,x,y)$.
We have $\delta_L \leq \delta(G) \leq \delta_L + 4k$.
\end{corollary}

\section{Extending Borassi et al. to $k$-domination} \label{sec:extending_borassi}

In this section, we build upon the algorithm proposed in~\cite{BorassiCCM15} to design a new algorithm for hyperbolicity that exploits $k$-domination.

\subsection{Skippable, acceptable and valuable.}
In this section, we show how to extend the pruning techniques proposed in~\cite{BorassiCCM15,CohenCL15} to the use of $k$-dominating sets.
As in~\cite{BorassiCCM15,CohenCL15}, we define for any quadruple $u,v,x,y$:
\begin{equation*}
\tau(u,v,x,y) \coloneqq \frac{1}{2}\big(\dd(u, v) + \dd(x,y)
- \max\{ \dd(x,u)+\dd(y,v), \dd(x,v)+\dd(y,u) \} \big).
\end{equation*}

As noted in~\cite{BorassiCCM15,CohenCL15}, we have $\delta(G) = \max_{u,v,x,y\in V} \tau(u,v,x,y)$. Indeed, if $\dd(u, v) + \dd(x,y)$ is the largest sum as defined in Definition~\ref{def:hyperbolic}, we have $\tau(u,v,x,y) = \delta(u,v,x,y)$, and otherwise $\tau(u,v,x,y)\leq 0$.

Given a $k$-dominating set $D$ of a connected graph $G$ and a quadruple $u,v,x,y \in D$, we now establish upper bounds on the values $\delta(u',v',x',y')$ and $\tau(u',v',x',y')$ for every quadruple $u',v',x',y'$ with respective associated dominators $u,v,x,y$.

\begin{lemma}[\cite{CohenCL15}]\label{lem:shape}
 For every quadruple $u,v,x,y$ of vertices of a connected graph $G$, we have $\delta(u,v,x,y)\leq \min_{a,b\in\set{u,v,x,y}}\dd(a,b)$. Furthermore, if $S_1 = \dd(u,v) + \dd(x,y)$ is the largest of the sums defined in Definition~\ref{def:hyperbolic} (which can be assumed w.l.o.g.), we have $\delta(u,v,x,y)\leq \frac{1}{2}\min\set{\dd(u,v),\dd(x,y)}$.
\end{lemma}

\begin{lemma}\label{lem:shape_dom}
	Given a $k$-dominating set $D$ of a connected graph $G$ and a quadruple $u,v,x,y \in D$, then for each quadruple $u',v',x',y'$ with respective associated dominators $u,v,x,y$,
	%such that $u'\in D^{-1}(u)$, $v'\in D^{-1}(v), x'\in D^{-1}(x)$ and $y'\in D^{-1}(y)$, 
	we have 
\begin{equation*}
\delta(u',v',x',y') 
\leq \min_{a,b\in\{u,v,x,y\}}k_{a}+k_{b}+\dd(a, b)
\leq 2k + \min_{a,b\in\{u,v,x,y\}}\dd(a, b)
\end{equation*}
Furthermore, assuming that $S_1 = \dd(u,v) + \dd(x,y)$ is the largest of the sums defined in Definition~\ref{def:hyperbolic} (which can be assumed w.l.o.g.),
for each quadruple $u',v',x',y'$ with respective associated dominators $u,v,x,y$ we have: 

$\delta(u',v',x',y')\leq k_{u}+k_{v} + k_{x}+k_{y} +\frac{1}{2}\min\set{\dd(u,v), \dd(x,y)} 
 \leq 4k + \frac{1}{2}\min\set{\dd(u,v), \dd(x,y)}$
 
 and
 
 $\tau(u',v',x',y') 
  \leq\frac{1}{2}\min\{k_{u}+k_{v}+\dd(u,v), k_{x}+k_{y}+\dd(x,y)\}
  \leq k + \frac{1}{2}\min\set{\dd(u,v), \dd(x,y)}$.

\end{lemma}

\begin{proof}
	%By \Cref{lem:shape}, we have that $\delta(u',v',x',y') \leq \min_{a,b\in\{u',v',x',y\}}\dd(a, b)$. Since for any pair $\{a,b\} \subset\{u',v',x',y\}$ such that $a\in D^{-1}(a')$ and $b\in D^{-1}(b')$ with $\{a',b'\} \subset\{u,v,x,y\}$, we have $\dd(a',b') - 2k \leq \dd(a,b) \leq \dd(a',b') + 2k$, by triangle inequality, we get $\delta(u',v',x',y') \leq 2k + \min_{a',b'\in\{u,v,x,y\}}\dd(a', b')$.
%
	By Lemma~\ref{lem:shape}, we have that $\delta(u',v',x',y') \leq \min_{a',b'\in\{u',v',x',y\}}\dd(a', b')$. Since for any pair $\{a',b'\} \subset\{u',v',x',y\}$ such that $D(a')=a$ and $D(b')=b$ with $\{a,b\} \subset\{u,v,x,y\}$, we have 
%$\dd(a',b') - 2k \leq$
$\dd(a',b') \leq \dd(a',a) + \dd(a,b) + \dd(b,b') \le k_{a}+k_{b} + \dd(a,b) $, by triangle inequality. We thus get $\delta(u',v',x',y') \leq \min_{a,b\in\{u,v,x,y\}}k_{a}+k_{b}+\dd(a, b)$.

Next, assuming that $S_1=\dd(u,v)+\dd(x,y)$ is the largest of the sums defined in Definition~\ref{def:hyperbolic}, %\lv{that $(u',v')$ and $(x',y')$ are far-apart,}
and that we have $\delta(u',v',x',y')>0$, we get by Lemmas~\ref{lem:dominating} and~\ref{lem:shape} that $\delta(u',v',x',y')\leq k_{u}+k_{v} + k_{x}+k_{y} + \delta(u,v,x,y) \leq k_{u}+k_{v} + k_{x}+k_{y}  + \frac{1}{2}\min\set{\dd(u,v), \dd(x,y)}$.

%Now, consider a quadruple $(u',v',x',y')$ dominated by $u,v,x,y$ and 
Define the sums $S'_1 = \dd(u',v')+\dd(x',y')$, $S'_2=\dd(u',x') + \dd(v',y')$ and $S'_3=\dd(u',y')+\dd(v',x')$.
We have $2\tau(u',v',x',y') = S'_1-\max\{S'_2, S'_3\} \leq S'_1-\frac{1}{2}(S'_2+S'_3)$ and $S'_2+S'_3= \dd(u',x') + \dd(v',y') + \dd(u',y')+\dd(v',x') = (\dd(u',x') + \dd(u',y')) + (\dd(v',x') + \dd(v',y'))$. By triangle inequality, we get $S'_2+S'_3 \geq 2\dd(x',y')$ and so $2\tau(u',v',x',y') \leq \dd(u',v') \leq k_u+k_v + \dd(u,v)$.
Similarly, we have by triangle inequality that $S'_2+S'_3= (\dd(u',x') + \dd(v',x')) + (\dd(u',y') + \dd(v',y')) \geq 2\dd(u',v')$, and so $2\tau(u',v',x',y') \leq \dd(x',y') \leq k_x+k_y + \dd(x,y)$.
%
%Now, consider a quadruple $(u',v',x',y')$ dominated by $u,v,x,y$ and the sums $S'_1 = \dd(u',v')+\dd(x',y')$, $S'_2=\dd(u',x') + \dd(v',y')$ and $S'_3=\dd(u',y')+\dd(v',x')$.
%We have $2\tau(u',v',x',y') = S'_1-\max\{S'_2, S'_3\} \leq S'_1-\frac{1}{2}(S'_2+S'_3)$ and $S'_2+S'_3= \dd(u',x') + \dd(v',y') + \dd(u',y')+\dd(v',x') = (\dd(u',x') + \dd(u',y')) + (\dd(v',x') + \dd(v',y'))$. By triangle inequality, we get $S'_2+S'_3 \geq 2\dd(x',y')$ and so $2\tau(u',v',x',y') \leq \dd(u',v') \leq d(u',u) + \dd(u,v) + \dd(v,v')=k_{u}+k_{v}+\dd(u,v)$.
%Similarly, we have by triangle inequality that $S'_2+S'_3= (\dd(u',x') + \dd(v',x')) + (\dd(u',y') + \dd(v',y')) \geq 2\dd(u',v')$, and so $2\tau(u',v',x',y') \leq \dd(x',y') \leq \dd(x',x) + \dd(x,y) + \dd(y,t)=k_{x}+k_{y}+\dd(x,y)$.
\end{proof}

Interestingly, the upper bound on $\tau(u',v',x',y')$ established in Lemma~\ref{lem:shape_dom} is significantly stronger than the upper bound on $\delta(u',v',x',y')$. This will help reducing the size of the search space, and in particular to decide when to stop the exploration of quadruples and to return the result.

By~\cite[Lemma 8]{BorassiCCM15}, we know that if $x,y,u \in V$ are such that $\ecc(u) + \dd(x,y) - \dd(x,u) - \dd(y,u) \leq 4\delta_L + 2$, where $\delta_L$ is a lower bound on the hyperbolicity of $G$, then for each $v\in V$, we have $\tau(u,v,x,y) \leq \delta_L$. Hence, for each such triple $(x,y,u)$, we can skip the exploration of the $n$ quadruples $(u,v,x,y)$ for $v\in V$. We now adapt this result to the case   $x,y,u \in D$.

\begin{lemma}\label{lem:cut_dom1}
	Given a $k$-dominating set $D$ of a graph $G$, a lower bound $\delta_L$ on the hyperbolicity of $G$ and the vertices $x,y,u\in D$ satisfying $2\ecc(u) + \dd(x,y) - \dd(x,u) - \dd(y,u) + K_8 \leq 4\delta_L + 1$ where $K_8=4k_{u}+2k_{x}+2k_{y}\le 8k$, then for each quadruple $u',v',x',y'$ with respective associated dominators $u,v,x,y$, we have $\tau(u',v',x',y') \leq \delta_L$. 
\end{lemma}

\begin{proof}
	Given $x,y,u \in D$ satisfying $2\ecc(u) + \dd(x,y) - \dd(x,u) - \dd(y,u) + K_8 \leq 4\delta_L + 1$, we assume for the sake of contradiction that there exists $x'\in D^{-1}(x)$, $y'\in D^{-1}(y)$, $u'\in D^{-1}(u)$ and $v'\in V$ such that $\delta_L< \tau(u',v',x',y')$.

Then, 
%\begin{align*}
%2\delta_L + 1 &\leq 2\tau(u',v',x',y') = \dd(u',v')+\dd(x',y') - \max\{\dd(u',x')+\dd(v',y'), \dd(u',y')+\dd(v',x')\}\\
%&\leq \dd(u',v')+\dd(x',y') - \frac{1}{2}\left(\dd(u',x')+\dd(v',y') + \dd(u',y')+\dd(v',x')\right)\\
%&\leq \dd(x',y') + \ecc(u') - \frac{1}{2}\left(\dd(u',x') + \dd(u',y') + \dd(x',y')\right)\\
%&\leq \ecc(u') - \frac{1}{2}\left(\dd(u',x') + \dd(u',y') -\dd(x',y')\right)
%\end{align*}
%
% DC: I do like that because it's too difficult to obtain nice looking
% equations with this template
%
$2\delta_L + 1 \leq 2\tau(u',v',x',y') = \dd(u',v') +\dd(x',y') - \max\{\dd(u',x')+\dd(v',y'), \dd(u',y')+\dd(v',x')\}
\leq \dd(u',v')+\dd(x',y') - \frac{1}{2}\left(\dd(u',x')+\dd(v',y') + \dd(u',y')+\dd(v',x')\right)
\leq \dd(x',y') + \ecc(u') - \frac{1}{2}\left(\dd(u',x') + \dd(u',y') + \dd(x',y')\right)
\leq \ecc(u') - \frac{1}{2}\left(\dd(u',x') + \dd(u',y') -\dd(x',y')\right)$.

Now, observe that %$\ecc(u)-k \leq$
we have $\ecc(u') \leq \ecc(u) +k_{u}$ by triangle inequality. Similarly, we have $ \dd(u',x') + \dd(u',y') -\dd(x',y') \geq \dd(u,x) -k_{u}-k_{x} + \dd(u,y)-k_{u}-k_{y} -\dd(x,y)-k_{x}-k_{y}$.
%Furthermore, we have 
%$\dd(u,x) + \dd(u,y) -\dd(x,y) - 6k\leq \dd(u',x') + \dd(u',y') -\dd(x',y')$.
%
Hence, we get $2\delta_L + 1  \leq \ecc(u) + k_{u} - \frac{1}{2}(\dd(u,x) + \dd(u,y) -\dd(x,y) - 2(k_{u}+k_{x}+k_{y}))$, and so $4\delta_L + 2  \leq 2 \ecc(u) + \dd(x,y) - \dd(u,x) - \dd(u,y) + K_8$ which contradicts the hypothesis. 
\end{proof}

Similarly, by~\cite[Lemma 9]{BorassiCCM15}, we know that if $x,y,u \in V$ are such that $\ecc(u) + \dd(x,y) - 3\delta_L -\frac{3}{2} \leq \max\{\dd(u,x), \dd(u,y)\}$, then for each $v\in V$, we have $\tau(u,v,x,y) \leq \delta_L$. Hence, for each such triple we can skip the exploration of $n$ quadruples. We now adapt this result to the case $x,y,u \in D$.

\begin{lemma}\label{lem:cut_dom2}
	Given a $k$-dominating set $D$ of a graph $G$, a lower bound $\delta_L$ on the hyperbolicity of $G$ and vertices $x,y,u\in D$ satisfying $\ecc(u) + \dd(x,y) - 3\delta_L - 1 + K_4 \leq \max\{\dd(u,x)-k_{x}, \dd(u,y) - k_{y}\}$ where 
	%$K_4= k_{u} + k_{v}+k_{x} + k_{y}\le 4k$, 
	{$K_4 = 2 k_{u} + k_{x} + k_{y}\le 4k$},
	then for each quadruple $u',v',x',y'$ with respective associated dominators $u,v,x,y$
        %such that $x'\in D^{-1}(x)$, $y'\in D^{-1}(y)$, $u'\in D^{-1}(u)$ and $v'\in V$,
        we have $\tau(u',v',x',y') \leq \delta_L$. 
\end{lemma}

\begin{proof}
% 	Given $x,y,u \in D$ satisfying $\ecc(u) + \dd(x,y) - 3\delta_L - 1 + 5k \leq \max\{\dd(u,x), \dd(u,y)\}$, we assume for the sake of contradiction that there exists $x'\in D^{-1}(x)$, $y'\in D^{-1}(y)$, $u'\in D^{-1}(u)$ and $v'\in V$ such that $\delta_L< \tau(u',v',x',y')$ or equivalently $2\delta_L + 1 \leq  2\tau(u',v',x',y')$.

% By \Cref{lem:shape}, we have $\dd(u',y') > \delta_L$, that is $\dd(u',y') \geq \delta_L + \frac{1}{2}$.
% Consequently, we get $2\delta_L + 1 \leq  2\tau(u',v',x',y') = \dd(u',v')+\dd(x',y') - \max\{\dd(u',x')+\dd(v',y'), \dd(u',y')+\dd(v',x')\} \leq \dd(u',v')+\dd(x',y') - \dd(u',x') - \dd(u',y') \leq \dd(x',y')+\ecc(v') - \dd(u',x') - \delta_L - \frac{1}{2}$. 

% Since $\ecc(v') \leq \ecc(v) + k_{v}$, $\dd(x', y') \leq \dd(x,y) + k_{x}+k_{y}$ and $\dd(u',x') \geq \dd(u,x) -k_{u} - k_{x}$, we obtain 
% $2\delta_L + 1 \leq \dd(x,y) + k_{x} + k_{y} + \ecc(v) + k_{v} - \dd(u,x) + k_{u} + k_{x} - \delta_L - \frac{1}{2}$, that is $\dd(u,x) - k_{x}\leq \dd(x,y) + \ecc(v) -3\delta_L - \frac{3}{2} + K_4$. 
% By exchanging the role of $x'$ and $y'$, we similarly get $\dd(u,y) - k_{y} \leq \dd(x,y) + \ecc(v) -3\delta_L - \frac{3}{2} +K$.
% These two inequalities contradict the hypothesis.
	Given $x,y,u \in D$ satisfying $\ecc(u) + \dd(x,y) - 3\delta_L - 1 + K_4 \leq \max\{\dd(u,x)-k_{x}, \dd(u,y) - k_{y}\}$, we assume for the sake of contradiction that there exists $x'\in D^{-1}(x)$, $y'\in D^{-1}(y)$, $u'\in D^{-1}(u)$ and $v'\in V$ such that $\delta_L< \tau(u',v',x',y')$ or equivalently $2\delta_L + 1 \leq  2\tau(u',v',x',y')$.\\
% --
By Lemma~\ref{lem:shape}, we have $\dd(u',y') > \delta_L$, that is $\dd(u',y') \geq \delta_L + \frac{1}{2}$.
Consequently, we get $2\delta_L + 1 \leq  2\tau(u',v',x',y') = \dd(u',v')+\dd(x',y') - \max\{\dd(u',x')+\dd(v',y'), \dd(u',y')+\dd(v',x')\} \leq \dd(u',v')+\dd(x',y') - \dd(u',x') - \dd(u',y') \leq \ecc(u') + \dd(x',y') - \dd(u',x') - \delta_L - \frac{1}{2}$. \\
% --
Since $\ecc(u') \leq \ecc(u) + k_{u}$, $\dd(x', y') \leq \dd(x,y) + k_{x}+k_{y}$ and $\dd(u',x') \geq \dd(u,x) -k_{u} - k_{x}$, we obtain 
$2\delta_L + 1 \leq \ecc(u) + k_{u} + \dd(x,y) + k_{x} + k_{y} - \dd(u,x) + k_{u} + k_{x} - \delta_L - \frac{1}{2}$, that is $\dd(u,x) - k_{x}\leq \ecc(u) + \dd(x,y) -3\delta_L - \frac{3}{2} + K_4$. 
By exchanging the role of $x'$ and $y'$, we similarly get $\dd(u,y) - k_{y} \leq \ecc(u) + \dd(x,y) -3\delta_L - \frac{3}{2} +K_4$.
These two inequalities contradict the hypothesis.
\end{proof}

We use Lemmas~\ref{lem:shape_dom}, \ref{lem:cut_dom1} and~\ref{lem:cut_dom2} to identify \emph{acceptable} and \emph{skippable} vertices. More precisely,
\begin{definition}\label{def:acceptable_dom}
Given a $k$-dominating set $D$ of a graph $G$, a lower bound $\delta_L$ on the hyperbolicity of $G$ and a pair $x,y\in D$ of vertices, we say that a vertex $u\in D$ is \emph{acceptable} if it satisfies all the following conditions. Otherwise, it is \emph{skippable}.
\begin{enumerate}
    \item $\dd(u,x) +k_{u}+k_{x} > \delta_L$ and $\dd(u,y) + k_{u}+k_{y}> \delta_L$ (Lemma~\ref{lem:shape_dom});
    \item $2\ecc(u) + \dd(x,y) - \dd(u,x) - \dd(u,y) + 4k_{u}+2k_{x}+2k_{y} > 4\delta_L + 2$ (Lemma~\ref{lem:cut_dom1});
    \item $\ecc(u) + \dd(x,y) - 3\delta_L -\frac{3}{2} + 2 k_{u} + k_{x} + k_{y} > \max\{\dd(u,x)-k_{x}, \dd(u,y)-k_{y}\}$ (Lemma~\ref{lem:cut_dom2}).
\end{enumerate}
\end{definition}

Thanks to Definition~\ref{def:acceptable_dom}, we can avoid exploring all quadruples for which either $u$ or $v$ is skippable, thus saving a significant amount of computations. To further reduce the search space, \cite{BorassiCCM15} introduce the notion of $c$-valuable vertices (or simply valuable vertices), where $c\in V$ is an arbitrarily chosen vertex, for instance a vertex with small centrality or eccentricity.

\begin{lemma}\label{lem:valuable}
Given a $k$-dominating set $D$ of a graph $G$, a pair $x,y\in D$ of vertices and any fixed node $c\in V$, define $f_c(z)= \frac{1}{2}\left( \dd(x,y) - \dd(x,z) - \dd(y,z)\right) + \dd(z,c) + 2k_z + k_{x} + k_{y}$ for any $z\in D$.  Then, for each pair $u,v \in D$ and each quadruple $u',v',x',y'$ with respective associated dominators $u,v,x,y$, we have $2\tau(u',v',x',y') \leq f_c(u) + f_c(v)$.
\end{lemma}

\begin{proof}
We have $2\tau(u',v',x',y')
= \dd(u',v')+\dd(x',y') - \max\{\dd(u',x')+\dd(v',y'), \dd(u',y')+\dd(v',x')\}
\leq \dd(u',v')+\dd(x',y') - \frac{1}{2}\left(\dd(u',x')+\dd(v',y') + \dd(u',y')+\dd(v',x')\right)$.
Furthermore, we have $\dd(u',v') \leq \dd(u',c) + \dd(v',c)$, $\dd(u',c) \leq \dd(u,c) + k_{u}$ and $\dd(v',c) \leq \dd(v,c) + k_{v}$.
Hence, we get $2\tau(u',v',x',y') \leq \dd(x,y)+\dd(u,c) + \dd(v,c) - \frac{1}{2}\left(\dd(u,x)+\dd(u,y) + \dd(v,x)+\dd(v,y) \right) + 2k_{u} + 2k_{v} + 2k_{x} + 2k_{y} = f_c(u) + f_c(v)$.
\end{proof}

As a consequence of Lemma~\ref{lem:valuable}, if $2\tau(u,v,x,y) + 2(k_{u} + k_{v} + k_{x} + k_{y}) > 2\delta_L$, then either $f_c(u) > \delta_L$ or $f_c(v) > \delta_L$.

\begin{definition}\label{def:valuable_dom}
Given a $k$-dominating set $D$ of a graph $G$, a lower bound $\delta_L$ on the hyperbolicity of $G$ and a pair $x,y\in D$ of vertices, we say that an acceptable vertex $u\in D$ is \emph{$c$-valuable} (or simply \emph{valuable}) if $f_c(u) > \delta_L$.
\end{definition}

Summarizing, if $\tau(u,v,x,y) > \delta_L$, then $u$ and $v$ are both acceptable and at least one of them is furthermore valuable.

Observe that, for a given pair $(x,y)\in D^2$, deciding if a vertex $u\in D$ is skippable, acceptable or valuable requires a constant number of operations, assuming that distances and eccentricities are known. We will discuss this aspect in \S~\ref{sec:implementation_notes}.

\subsection{Algorithm.}

We can now present an exact algorithm (Algorithm~\ref{alg:hyp_dom_DFS}) that exploits the notion of $k$-domination to prune the search space and significantly reduce the memory usage compared to the algorithms proposed in~\cite{BorassiCCM15,CohenCL15,coudert:hal-03201405}.
%Algorithm~\ref{alg:hyp_dom_DFS} takes as input a connected graph $G$, a domination distance $k\geq 0$ and a reduction coefficient $r>1$. 
Algorithm~\ref{alg:hyp_dom_DFS} takes as input a connected graph $G$ and a sequence $k_i, k_{i-1},\cdots, k_0$ of domination distances, with $i\geq 0$ and $k_i = k > k_{i-1} >  \cdots > k_0 = 0$.

The general principle of the algorithm is to first scan the quadruples of $D_i$ {where $D_i$ is a $k_i$-dominating set}. Then, thanks to Lemma~\ref{lem:dominating}, we know that if a quadruple $u,v,x,y\in D_i$ is such that 
$\delta(u,v,x,y) +4k > \delta_L$, where $\delta_L$ is the current lower bound on $\delta(G)$, it may dominate a quadruple $u',v',x',y'$ such that $\delta(u',v',x',y') > \delta_L$.
In this case, instead of exploring directly all the quadruples of $V$ that are dominated by $u,v,x,y$, we scan the quadruples made of the nodes of a $k_{i-1}$-dominating set $D_{i-1}$ that are dominated by $u,v,x,y$. %We proceed similarly each time a quadruple of $D_j$, with $0< j \leq i$, satisfies the conditions of Lemma~\ref{lem:dominating}.
{We proceed similarly for smaller and smaller domination radii $k_j$ with $0<j<i$ as long as the condition of Lemma~\ref{lem:dominating} is satisfied for $k_j$.}
When a quadruple does not satisfies the conditions of Lemma~\ref{lem:dominating}, we prune the exploration of all dominated quadruples. We also use Lemma~\ref{lem:shape_dom} and Definitions~\ref{def:acceptable_dom} and~\ref{def:valuable_dom} to prune more quadruples, {generalizing in some extend the approach of~\cite{BorassiCCM15}}.
%More precisely, method \computeAccVal($x,y,D_i$) scans the vertices of $D_i$ and identifies those which are acceptable and/or valuable for the pair $(x,y)$ using available distance matrix 

Let us now present the algorithm in details.

%The parameters $k$ and $r$ are used to compute a \emph{hierarchical dominating set} of $G$ (line~\ref{line:DFS:HDS} of Algorithm~\ref{alg:hyp_dom_DFS}).
The first step of the algorithm is the construction of a hierarchical dominating set (line~\ref{line:DFS:HDS} of Algorithm~\ref{alg:hyp_dom_DFS}). This is done as follows. 
%Firstly, we compute distances $k_i=k$, $k_{i-1}=\PartIntInf{k_i/r}$, $k_{i-2}=\PartIntInf{k_{i-1}/r}$, \ldots, $k_0=0$, where the index $i$ is the smallest possible.
We compute a $k_i$-dominating set $D_i$ for $G$, using a greedy algorithm (see \S~\ref{sec:implementation_notes}). We associate to each vertex $v$ the closest vertex $D_i(v)$ that dominates it and let $V_{i,u}=D_i^{-1}(u)$ denote the set of vertices associated to $u$ for each $u\in D_i$. Note that $\{V_{i,u} : u\in D_i\}$ forms a partition of $V$. We let $k_{i,u}=\max_{v\in V_{i,u}}\dd(u,v)$ denote the domination radius of $u\in D_i$. Obviously, we have $k_{i,u} \leq k_i$. We proceed similarly for each $j < i$ as follows. Given a subset $V'\subseteq V$ we define a \emph{$k$-dominating set} of $V'$ as a set $D\subseteq V'$ such that each vertex $v\in V'$ is dominated by some vertex $u\in D$ such that $\dd(u,v)\le k$. Given a partition of $V$ into sets $V_{j+1,u}$ for $u\in D_{j+1}$, we compute a $k_j$-dominating set $D_{j,u}$ of $V_{j+1,u}$ for each $u\in D_{j+1}$. Each vertex $v'\in V_{j+1,u}$ is associated to its closest dominator $D_{j,u}(v')$ in $D_{j,u}$. We then let $V_{j,u'}=D_{j,u}^{-1}(u')$ denote the set of vertices associated to $u'$ for each $u'\in D_{j,u}$ and let $k_{j,u'}=\max_{v'\in V_{j,u'}}d(u',v')$ denote its domination radius. We then obtain a new partition of $V$ as $\{V_{j,u'} : u'\in D_j\}$ where $D_j=\cup_{u\in D_{j+1}}D_{j,u}$. As $k_0=0$, this process ends with $D_{0,u}=V_{1,u} = \{u\}$ for each $u\in D_1$ and $D_0=V$. 
%The next step is to compute, for each vertex $u\in D_i$, a $k_{i-1}$-dominating set $D_{i-1,u}$ of the subgraph $G[D_{i,u}^{-1}]$ induced by the vertices with associated dominator $u$ and similarly associated each node in $D_{i,u}$ to its closest dominator. The union of the resulting dominating sets forms $D_{i-1}$. We proceed similarly for each vertex $u'\in D_{i-1}$ and the graph induced by the vertices with associated dominator $u'$. This process ends with $D_0=V$. For convenience, we let $D_{j,u}$ denote the vertices with associated dominator $u\in D_j$ %\lv{$u\in D_j$ ? Conflict of notation with $D(u)$, $D_j^{-1}(u)$ is not good either as it does take into account balls.} %and by $H_j(u)$ the vertices of the $k_{j-1}$-dominating set of $G[D_j(u)]$, for each $i\geq j\geq 0$ and $u\in D_j$.

%See \S~\ref{sec:hierarcicalDominatingSet} for more details.

Next, Algorithm~\ref{alg:hyp_dom_DFS} computes the distances between all pairs of vertices of the $k$-dominating set $D_i$ and stores these distances in matrix $\dd$ (line~\ref{line:DFS:distancematrix}).  The construction of $\dd$ can be done in time $\OO(|D_i|(n + m))$ using BFS and uses spaces $\OO(|D_i|^2)$. We assume here that $|D_i|$ is small enough to ensure that matrix $\dd$ can fit into memory.

Algorithm~\ref{alg:hyp_dom_DFS} uses the notion of mates, introduced in~\cite{BorassiCCM15}. Roughly, $v$ is a mate of $u$ if $\dd(u, v) \geq \dd(x, y)$ and the pair $(u, v)$ has been considered before the pair $(x,y)$ that is currently considered. Hence, when all quadruples involving a pair $(x,y)$ have been considered, $x$ becomes a mate of $y$ and $y$ becomes a mate of $x$. This notion is used line~\ref{line:DFS:for_mates} of the algorithm to reduce the number of pairs $(u, v)$ to consider. The maintenance of the mates requires overall time and space in $\OO(|D_i|^2)$.

The main part of the algorithm is lines~\ref{line:DFS:for_xy} to~\ref{line:DFS:update_mates} which uses the best value found so far ($\delta_L$) and the bounds of Lemma~\ref{lem:shape_dom} to stop computations as soon as possible. More precisely, the algorithm considers all pairs $(x,y)$ of vertices in $D_i$ sorted by non-increasing distance. This ordering is a consequence of Lemma~\ref{lem:shape_dom}. Indeed, if a pair $(x,y)$ is such that $\dd(x,y) + 2k_i \leq 2\delta_L $, then we also have $\dd(x',y') + 2k_i \leq 2\delta_L $ for all pairs $(x',y')$ such that $\dd(x',y') \leq \dd(x,y)$.
Hence, we can stop computation and return the result as soon as the condition is satisfied (line~\ref{line:DFS:stop_shape}).
Furthermore, Lemma~\ref{lem:shape_dom} also ensures that for a pair $(x,y)$ such that $\dd(x,y) + k_{i,x} + k_{i,y} \leq 2\delta_L$, we have $\tau(u, v, x, y) \leq \delta_L$ for all pairs $(u,v)$, and so we can skip computations involving this pair $(x,y)$ (line~\ref{line:DFS:skip_xy}).

Then, the algorithm determines, through a call to method \computeAccVal($x,y,D_i$), the acceptable and valuable vertices of $D_i$ for pair $(x, y)$  according Definitions~\ref{def:acceptable_dom} and~\ref{def:valuable_dom} (line~\ref{line:DFS:accval}). 
It simply consists in a linear scan of $D_i$, as explained in \S~\ref{sec:implementation_notes}.
The algorithm then considers all quadruples $(u, v, x, y)$ made of the pair $(x, y)$, a valuable vertex $u$ and an acceptable vertex $v$ that is also a mate of $u$.
The algorithm first checks whether this quadruple enables to improve the lower bound $\delta_L$ (line~\ref{line:DFS:improve_LB}). Then, it checks using Lemma~\ref{lem:dominating} whether this quadruple may dominates another quadruple that could help improving the bound (lines~\ref{line:DFS:explore1}-\ref{line:DFS:explore2}). If this is the case, it calls method \texttt{explore} (Algorithm~\ref{alg:hyp_dom_DFS_rec}) that we  describe next.

\begin{algorithm}
\caption{Compute hyperbolicity using $k$-domination.}
\label{alg:hyp_dom_DFS}
\begin{algorithmic}[1]
\Require $G = (V, E)$, a connected graph
\Require $k_i = k > k_{i-1} > \cdots > k_0=0$, a sequence of domination distances
%\Require $k\geq 0$, the initial domination distance
%\Require $r > 1$, reduction coefficient
%\State Let $k_i=k$, $k_{i-1}=\PartIntInf{k_i/r}$, $k_{i-2}=\PartIntInf{k_{i-1}/r}$, \ldots, $k_0=0$
\State ${\cal H} = \{(D_j,(D_{j-1,u})_{u\in D_j},k_j),\ 0\leq j\leq i\} \gets \HRS(G, (k_i,k_{i-1},\ldots,0))$ \label{line:DFS:HDS}
\State Let $\dd$ be the distance matrix between the vertices in $D_i$ \label{line:DFS:distancematrix}
\State $\delta_L \gets 0$
\State \mates$[v] \gets \emptyset$ for each $v\in D_i$ \label{line:DFS:init_mates}
\ForAll{$(x,y) \in D_i^2$ sorted by non increasing distances \label{line:DFS:for_xy}}
    \If{$\dd(x,y) + 2k_i \leq 2\delta_L $} \label{line:DFS:stop_shape}
        \State \textbf{return} $\delta_L$
    \EndIf
    \If{$\dd(x,y) + k_{i,x} + k_{i,y} \leq 2\delta_L $} \label{line:DFS:skip_xy}
        \State continue
    \EndIf
    \State (\texttt{acceptable}, \texttt{valuable}) $\gets$ \label{line:DFS:accval}% continue on next line
    %\Statex \hfill
    \computeAccVal($x,y,D_i$) 
    \For{$u \in $ \upshape\texttt{valuable}}
        \For{$v \in $ \upshape\texttt{mates}[$u$]} \label{line:DFS:for_mates}
            \If{$v \in$ \upshape\texttt{acceptable}}
                \State $\delta_L\gets\max\{\delta_L, \tau(u,v,x,y)\}$ \label{line:DFS:improve_LB}
                \State $\delta_{\textsc{ub}}\gets \tau(u,v,x,y) + \sum_{z\in\{u,v,x,y\}} k_{i,z}$
%                \If{$\tau(u,v,x,y) + k_{i,u} + k_{i,v} + k_{i,x} + k_{i,y} > \delta_L$ \label{line:DFS:explore1}}
                \If{$\delta_{\textsc{ub}} > \delta_L$ \label{line:DFS:explore1}}
                    \State $\delta_L\gets \max\{\delta_L, $   \label{line:DFS:explore2}
                    %\Statex \hfill
                    $\texttt{explore}(u,v,x,y, {\cal H}, i-1)\}$  
                \EndIf
            \EndIf
        \EndFor
    \EndFor
    \State add $y$ to \texttt{mates}[$x$] 
    \State add $x$ to \texttt{mates}[$y$] \label{line:DFS:update_mates}
\EndFor
\State \textbf{return} $\delta_L$
\end{algorithmic}
\end{algorithm}

Given a quadruple $(u,v,x,y)$ of vertices in $D_{j+1}$,
%($j+1 = i$ when called from Algorithm~\ref{alg:hyp_dom_DFS}),
method \texttt{explore} (Algorithm~\ref{alg:hyp_dom_DFS_rec}) considers all pairs $(x', y')$ such that $x'\in D_{j,x}$ and $y' \in D_{j,y}$. Recall that $D_{j, x}$ is the set of vertices of a distance $k_{j-1}$ dominating set of the set of vertices $V_{j,x}$ that are dominated by $x$ in a distance $k_j$ dominating set.
Thanks to Lemma~\ref{lem:shape_dom}, the algorithm do not consider the pairs $(x', y')$ that can not help improving the lower bound (line~\ref{line:DFSrec:skip}).

Next, Algorithm~\ref{alg:hyp_dom_DFS_rec} determines the subsets of vertices of $D_{j,u}$ and $D_{j,v}$ that are acceptable or valuable for $(x',y')$ according Definitions~\ref{def:acceptable_dom} and~\ref{def:valuable_dom} (lines~\ref{line:DFSrec:acc1} and~\ref{line:DFSrec:acc2}). 
Then, it considers the quadruples $(u',v',x',y')$ such that either $u'$ is a valuable vertex of $D_{j,u}$ and $v'$ an acceptable vertex of $D_{j,v}$, or $u'$ is an acceptable vertex of $D_{j,u}$ and $v'$ a valuable vertex of $D_{j,v}$.
For each quadruple $(u',v',x',y')$, it checks using Lemma~\ref{lem:dominating} whether this quadruple may dominate another quadruple that could help improving the bound (line~\ref{line:DFSrec:rec1}). If this is the case, it performs a recursive call to \texttt{explore} to consider the dominated quadruples in a distance $k_{j-1}$ dominating set (line~\ref{line:DFSrec:rec2}).

\begin{algorithm}[t]
\caption{Method {\upshape\texttt{explore}}$(u,v,x,y,{\cal H},j)$.}
  \label{alg:hyp_dom_DFS_rec}
\begin{algorithmic}[1]
\Require $G = (V, E)$, a connected graph
\Require $j$, the current index for domination
\Require $u,v,x,y\in D_{j+1}$, a quadruple of $G$
\Require ${\cal H}= \{(D_{j'},(D_{j'-1,u})_{u\in D_{j'}},k_{j'}),\ 0\leq j'\leq i\}$, hierarchical dominating set of $G$
\If{$j = -1$} % j == 0
    \State \textbf{return} $\tau(u,v,x,y)$
\EndIf
\For{$x'\in D_{j,x}$}
    \For{$y' \in D_{j,y}$}
        \If{$\dd(x',y') + k_{j-1,x'} + k_{j-1,y'} \leq 2\delta_L $} \label{line:DFSrec:skip}
            \State continue
        \EndIf
        \State (\texttt{acceptable}$(u)$, \texttt{valuable}$(u)$) $\gets$ \label{line:DFSrec:acc1}
        %\Statex \hfill
        \computeAccVal($x', y', D_{j,u}$) 
        \State (\texttt{acceptable}$(v)$, \texttt{valuable}$(v)$) $\gets$ \label{line:DFSrec:acc2}
        %\Statex \hfill
        \computeAccVal($x', y', D_{j,v}$) 
        \For{$(a, b) \in \{(u,v), (v,u)\}$}
            \For{$u'\in $ \upshape\texttt{valuable}$(a)$}
                \For{$v'\in $ \upshape\texttt{acceptable}$(b)$}
                    \State $\delta_L\gets\max\{\delta_L, \tau(u',v',x',y')\}$
                    \State $k_4 \gets \sum_{z'\in\{u',v',x',y'\}} k_{j-1,z'}$
                    \If{$\tau(u',v',x',y') + k_4 > \delta_L$ \label{line:DFSrec:rec1}}
                        \State $\delta_L \gets \max\{ \delta_L,$ \label{line:DFSrec:rec2}
                        %\Statex \hfill
                        \texttt{explore}$(u',v',x',y', {\cal H}, j-1) \}$
                    \EndIf
                \EndFor
            \EndFor
        \EndFor
    \EndFor
\EndFor
\State \textbf{return} $\delta_L$
\end{algorithmic}
\end{algorithm}

\begin{observation} \em % to make the text normal
When the domination distance is zero ($k = 0$), the test on line~\ref{line:DFS:explore1} of Algorithm~\ref{alg:hyp_dom_DFS} will always be false, and so the algorithm is almost the one proposed in~\cite{BorassiCCM15}. Indeed, the only difference is that the algorithm of~\cite{BorassiCCM15} considers the list of \emph{far-apart pairs} instead of the list of all pairs, which helps further pruning the search space.
Roughly, $u$ is far from $v$ if $\dd(u,v) \geq \dd(u, w)$ for all $w\in N(v)$, and the pair $(u, v)$ is far-apart if $u$ is far from $v$ and $v$ is far from $u$. It has been proved in~\cite{SotoGomez2011} that any graph contains two far-apart pairs $(u,v)$ and $(x,y)$ satisfying $\delta(u,v,x,y) = \delta(G)$.
%it is always possible to find a quadruple $(u,v,x,y)$ such that $\delta(u,v,x,y) = \delta(G)$ and both pairs $(u,v)$ and $(x,y)$ are far-apart. 
An interesting open question is therefore how to use the notion of far-apart pairs in Algorithms~\ref{alg:hyp_dom_DFS} and~\ref{alg:hyp_dom_DFS_rec} to further prune the search space.
\end{observation}

\begin{observation} \em % to make the text normal
Since $\max_{a',b'\in V_{i,u}} \dd(a',b') \leq 2k$ for all $u\in D_i$, we know by Lemma~\ref{lem:shape_dom} that, when $\delta_L\geq 2k$,  Algorithm~\ref{alg:hyp_dom_DFS} has to consider only quadruples of $D_i$ composed of four distinct nodes. However, when $\delta_L < 2k$, it is necessary to extend the list of pairs $(x,y)$ to consider line~\ref{line:DFS:for_xy} of Algorithm~\ref{alg:hyp_dom_DFS} by the list of pairs $\{(x,x) : x\in D_i\}$ in order to also consider quadruples of $D_i$ with two or more identical nodes.
%. This way, the algorithm will also consider quadruples of $D_i$ with two or more identical nodes.
\end{observation}

\subsection{Implementation notes.} \label{sec:implementation_notes}

We now discuss how to efficiently implement Algorithms~\ref{alg:hyp_dom_DFS} and~\ref{alg:hyp_dom_DFS_rec}.

Let us start with the method used to create the dominating sets. We use a greedy approach where we order the nodes decreasingly according to their degrees and then, in this order, add a node to the dominating set if and only if it is not dominated yet for the considered radius. During this process, we ensure that a dominated vertex is attached to the closest vertex that dominates it, and we record for each vertex of the dominating set its domination radius (i.e., largest distance to a vertex it dominates). While we tried different strategies to construct dominating sets (e.g., order the vertices by increasing/decreasing degree/eccentricity, ensure that the vertices dominated by a vertex induce a connected subgraph, etc.), we could not observe large performance changes on the hyperbolicity computation. In particular, we do not believe that this part of the algorithm is worthwhile optimizing further.

Let us now consider the data used in Algorithm~\ref{alg:hyp_dom_DFS}, and in particular for the classification of vertices as skippable, acceptable or valuable. This step of the algorithm requires the knowledge of the eccentricities of the vertices (Definition~\ref{def:acceptable_dom}), of the distances from a central vertex (Lemma~\ref{lem:valuable}), of the effective domination distance of each considered vertex (given by the hierarchical dominating set), and of the distances between the considered vertices.
Firstly, it is possible to efficiently compute the eccentricities of the vertices using the algorithms proposed in~\cite{TakesK13,DraganHV18,LiQQ+2018} that perform a smart management of the upper and lower bounds on the eccentricity of each vertex in order to avoid computing distances from each vertex. In our implementation, we have chosen to use the algorithm proposed in~\cite{DraganHV18}.
Secondly, as noticed in~\cite{BorassiCCM15}, a good choice for the central vertex is a vertex with small eccentricity or farness (sum of the distances from a vertex to all other vertices). Since we have already computed the eccentricities of the vertices, we choose a vertex with minimum eccentricity as central vertex, and we store distances from that central vertex (computed using BFS).
Now, observe that Algorithm~\ref{alg:hyp_dom_DFS} uses the distance matrix between all the vertices in $D_i$ (line~\ref{line:DFS:distancematrix}). This matrix of distances is obtained solving a BFS from each vertex in $D_i$. Consequently, each required data for the classification of the vertices (method \computeAccVal, line~\ref{line:DFS:accval} of Algorithm~\ref{alg:hyp_dom_DFS}) can be accessed in constant time and so this classification requires a constant number of operations per vertex.

% An important ingredient of \Cref{alg:hyp_dom_DFS,alg:hyp_dom_DFS_rec} is the classification of vertices as skippable, acceptable or valuable, and this step of the algorithms requires the knowledge of the eccentricities of the vertices (\Cref{def:acceptable_dom}), of the distances from a central vertex (\Cref{lem:valuable}), of the effective domination distance of each considered vertex, and of the distances between the considered vertices.

Let us continue with the data used in Algorithm~\ref{alg:hyp_dom_DFS_rec}. Apart from the eccentricities of the vertices, the distances from the central vertex and the effective domination distances that have already been computed for Algorithm~\ref{alg:hyp_dom_DFS}, we need the distances between the vertices dominated by the quadruple $u,v,x,y$. Similarly, method \computeAccVal{} linearly scans all vertices $u'$ in a dominating set and requires distances between $u'$ and a pair $x',y'$ of nodes to determine whether $u'$ is acceptable and/or valuable.
In order to get all these distances, we use a hub labeling (also called 2-hop labeling) of the graph~\cite{GavoillePPR04}. Roughly, this data structure assigns to each vertex $u$ of the graph a \emph{label}, that is list of vertices with shortest path distance from $u$ to each of these vertices. The construction of the lists ensures that for each pair of vertices $u,v\in V$, the intersection of the lists of $u$ and $v$ contains a vertex that is on a shortest path between $u$ and $v$. Then, to get the distance between $u$ and $v$, it suffices to search for the vertex minimizing the distance in the intersection of the lists. In practice, the labels are shorts (i.e., of poly-logarithmic length in the size of the graph) and so distance queries can be answered efficiently for graphs with millions of nodes. Furthermore, this data structure uses little space. Practical algorithms for building such data structure have time complexity in $\OO(nm)$, and we choose to use the algorithm proposed in~\cite{AkibaIY13} and~\cite{DellingGPW14} that combines several methods to significantly speedup the construction.

Now, since the distance between say $u'\in D_{j}^{-1}(u)$ and $v'\in D_{j}^{-1}(v)$ is used several times in Algorithm~\ref{alg:hyp_dom_DFS_rec}, we first store these distances in a matrix, thus enabling to access these distances in constant time when needed, and so to classify a vertex in constant time. More precisely, at the beginning of Algorithm~\ref{alg:hyp_dom_DFS_rec}, we build for each pair $(u,v)$, $(x,u)$, $(y,u)$, $(x,v)$, $(y,v)$ and $(x,y)$ a rectangular matrix storing the required distances (e.g., all distances between $u'\in D_{j}^{-1}(u)$ and $v'\in D_{j}^{-1}(v)$). We thus query the hub labeling data structure only once per distance.
In order to further reduce the number of queries to the hub labeling data structure, we use a cache of distance matrices. Indeed, two consecutive calls to Algorithm~\ref{alg:hyp_dom_DFS_rec} for domination index $j$ might involve a same pair $(x,y)$, and possibly a same triple $u,x,y$. This cache is implemented as a bounded size doubly linked list of matrices along with a mapping from pairs of vertices to elements of this list. This way, if the matrix needed for a pair of vertices is in the cache, we move it to the end of the doubly linked list and return it. Otherwise, we build a new matrix, insert it at the end of the list and record it in the mapping. When the maximum size of the cache is reached, we remove the first element of the list and the corresponding entry from the mapping, and then create a new matrix and insert it at the end of the list. Since we use only 7 matrices at the same time (the 6 matrices listed above and the matrix for the pair $(v, u)$), this method is safe if the size of the cache is at least 7. In practice the matrices are small and so we can store millions of them.
However, if we observe that a matrix is too large to be stored (e.g., if its side is more than 50\,000), we avoid the construction of the matrices and directly query the hub labeling data structure.

Finally, in order to further reduce the overall running time of the algorithm, we perform two pass of lines~\ref{line:DFS:init_mates}-\ref{line:DFS:update_mates} of Algorithm~\ref{alg:hyp_dom_DFS}. During the first pass, we omit lines~\ref{line:DFS:explore1}-\ref{line:DFS:explore2}, and so the calls to \texttt{explore}, and we obtain a good lower bound $\delta_L$ such that $\delta_L \leq \delta(G) \leq \delta_L + 4k$. 
Then, during the second pass of lines~\ref{line:DFS:init_mates}-\ref{line:DFS:update_mates} of Algorithm~\ref{alg:hyp_dom_DFS},  we call \texttt{explore} on the quadruples that can help improving the lower bound $\delta_L$. Since we start the second pass with a better lower bound $\delta_L$, we prune many calls to \texttt{explore}.
In addition, we also record for each pair $(x,y)$ the best found value $\tau(u,v,x,y) + k_{i,u} + k_{i,v} + k_{i,x} + k_{i,y}$ over all considered pairs $(u,v)$ of valuable and acceptable vertices during the first pass. This value is in fact an upper bound on the value $\tau(u',v',x',y')$ of any quadruple $u',v',x',y'$ having respective associated dominators $u,v,x,y$ such that $\dd(u,v) \geq \dd(x,y)$. Hence, during the second pass, we can avoid considering the pairs $(x,y)$ for which this value is less than the current lower bound $\delta_L$, and so that cannot lead to any call to \texttt{explore}.

\subsection{Correctness of the algorithm.}

The validity of the algorithms mainly follows from Lemmas~\ref{lem:dominating} and~\ref{lem:shape_dom}, and Definitions~\ref{def:acceptable_dom} and~\ref{def:valuable_dom}. 
More precisely, the validity of the pruning at line~\ref{line:DFS:explore1} in Algorithm~\ref{alg:hyp_dom_DFS} comes from Lemma~\ref{lem:dominating}. The general proof of correctness is similar to that of~\cite{BorassiCCM15} using Definitions~\ref{def:acceptable_dom} and~\ref{def:valuable_dom}, and relying on Lemmas~\ref{lem:shape_dom}, \ref{lem:cut_dom1} and~\ref{lem:cut_dom2}.
Observe that a quadruple $u,v,x,y$ may be considered several times. Indeed, it dominates itself, and so can be considered with a smaller radii during the recursive calls to method \texttt{explore}. 
However, a quadruple is considered at most $i+1$ times (i.e., the number of radii) and so the worst case time complexity is $\OO(n^4)$ when $i$ is a small constant.
The space usage is dominated by the hub labeling and the distance matrix for nodes in $D_i$. Although this could be $\OO(n^2)$ in the worst case, it is close to linear when $|D_i|=\OO(\sqrt{n})$ and the average hubset size is small.

%%%%%%%%%%%%%%%%%%%%%%%%%%%%%%%%%%%%%%%%%%%%%%%%%%%%%%
%%%%%%%%%%%%%%%%%%%%%%%%%%%%%%%%%%%%%%%%%%%%%%%%%%%%%%
%%%%%%%%%%%%%%%%%%%%%%%%%%%%%%%%%%%%%%%%%%%%%%%%%%%%%%

\begin{table*}
\caption{Selected graph parameters of all the graphs that we use in our experiments. Note that for all graphs, we extracted the largest biconnected component and restrict to this subgraph in our experiments.}
\label{tab:graphs}
\centering
\begin{tabular}{|l|cc|ccc|ccc|}
\cline{4-9}
\multicolumn{3}{l|}{} & \multicolumn{3}{c|}{\textbf{eccentricity}} & \multicolumn{3}{c|}{\textbf{degree}} \\
\hline
\textbf{Graph} & \textbf{\#nodes} & \textbf{\#edges} & \textbf{radius} & \textbf{mean} & \textbf{diameter} &  \textbf{min} & \textbf{mean} & \textbf{max}\\
\hline
\hline
\texttt{NotreDame} & 134\,958 & 833\,732 & 18 & 20.99 & 36 & 2 & 12.36 & 10\,721\\
\texttt{web-Stanford} & 181\,906 &  1\,676\,077 & 46 & 48.81 & 92 & 2 & 18.43 & 35\,488 \\
\texttt{web-BerkStan} & 489\,296 & 5\,939\,242 & 46 & 49.19 & 92 & 2 & 24.27 & 80\,733\\
\hline
\texttt{t.CAL} & 1\,267\,004 & 1\,671\,989 & 1\,149 & 1\,656.89 & 2\,298 & 2 & 2.64 & 7 \\
\texttt{t.FLA} & 691\,175 & 941\,893 & 890 & 1\,378.52 & 1\,780 & 2 & 2.73 & 8\\
\texttt{roadNet-PA} & 863\,105 & 1\,313\,732 & 401 & 689.81 & 793 & 2 & 3.04 & 9 \\
\hline
\texttt{buddha} & 543\,652 & 1\,631\,574 & 244 & 360.44 & 487 & 3 & 6.0 & 17\\
\texttt{froz} & 749\,520 & 2\,895\,228 & 812 & 1\,130.38 & 1\,451 & 2 & 7.73 & 8\\
\hline
\texttt{grid300-10} & 90\,211 & 162\,152 & 300 & 450.50 & 600 & 2 & 3.6 & 4\\
\texttt{grid500-10} & 250\,041 & 449\,831 & 500 & 750.49 & 1\,000 & 2 & 3.6 & 4\\
\texttt{z-alue7065} & 34\,040 & 54\,835 & 213 & 319.43 & 426 & 2 & 3.22 & 4\\
\hline
\end{tabular}
\end{table*}

\section{Experimental Evaluation}\label{sec:experimentation}

We now turn to experimentally evaluating the algorithm presented in \S~\ref{sec:extending_borassi}. To this end, we first explain the setup of the experiments to then compare our approach to the previous state of the art. As our algorithm depends on input parameters, we subsequently evaluate how the choice of parameters affects the running time.

\subsection{Setup.}

For our experimental evaluation, we used 
web graphs (\texttt{NotreDame}, \texttt{web-BerkStan}, \texttt{web-Stanford}),
road networks (\texttt{t.CAL}, \texttt{t.FLA}, \texttt{roadNet-PA}),
a 3D triangular mesh (\texttt{buddha}), 
a grid-like graphs from VLSI applications (\texttt{z-alue7065}) 
and a graph from a computer game (\texttt{froz}).
The data is available from 
\url{snap.stanford.edu}, 
%\url{webgraph.di.unimi.it}, 
\url{www.dis.uniroma1.it/challenge9}, 
\url{graphics.stanford.edu},\break
\url{steinlib.zib.de}, 
and \url{movingai.com}. 
We also use synthetic graphs (\texttt{grid300-10}, \texttt{grid500-10}) which are square grids with respective sides $301$ and $501$ where 10\% of the edges have been randomly deleted.
Each graph is taken as an undirected unweighted graph and we consider only its largest biconnected component.
See Table~\ref{tab:graphs} for the characteristics of these graphs.
The chosen graphs have a large number of nodes compared to the graph sizes that were feasible for previous implementations that compute the graph hyperbolicity. Furthermore, our approach relies on the fact that pruning of quadruples is actually possible. Thus, the graphs that we consider do not exhibit a very low hyperbolicity --- the lowest is 8 (\texttt{NotreDame}).

We implemented our algorithms in C++ and our code can be found at~\cite{gitlabv2}.
For all except two experiments, we used a computer equipped with Intel Xeon Gold 6240 CPUs operating at 2.6GHz and 192GB RAM. The other two experiments (the very long experiments in Table~\ref{tab:experiments}) were run on a computer equipped with an Intel Core i9-10900K CPU operating at 3.7GHz and 64GB RAM. Both machines run Linux.
Furthermore, all computations were conducted using a single thread.

For the hub labeling, we use the algorithm proposed in~\cite{AkibaIY13}, except for road networks for which we use~\cite{DellingGPW14} as we observed that it produces shorter labels in this case. 
Reported running times have been measured using \texttt{/usr/bin/time -v} and include all steps of the program, from reading data from file to returning the result. Memory limits have been set to 192GB using \texttt{ulimit}.
%\AN{I had a look and am fine with it. I also added that we are using linux in the previous paragraph.}

\subsection{Parameters.} \label{sec:exp_parameters}

In Algorithm~\ref{alg:hyp_dom_DFS_rec}, we assume that a hierarchical dominating set is given as input. Especially, we require the domination distances to be given as parameters. Note that the number of possible sequences of domination distances is huge and it is impractical as parameter of an algorithm. We thus reduce the potential sequences by just requiring two parameters: the largest domination distance and the ratio by which it is reduced in each round. More formally, given a maximal domination distance $k$ and a ratio $r$, the sequence of domination distances $k_i, \dots, k_0$ of the dominating sets then is defined as
\[
	k_i = k,\quad k_{i-1} = \left\lfloor\frac{k_i}{r}\right\rfloor,\quad \dots ,\quad k_0 = \left\lfloor\frac{k_{1}}{r}\right\rfloor = 0.
\]
Note that the value of $i$ is implicitly given by the number of steps until this process reaches zero.

\subsection{Comparison with Previous Work.}

To evaluate the improvement of our algorithm over previous work, we compare to the implementation of~\cite{coudert:hal-03201405}, which was shown to outperform --- considering the trade-off between running time and memory consumption --- the algorithm of~\cite{BorassiCCM15}. See~\cite[Figures~1~and~2]{coudert:hal-03201405} for a comparison of the running time and memory consumption of~\cite{BorassiCCM15} and~\cite{coudert:hal-03201405}.
In particular, the memory consumption of~\cite{BorassiCCM15} is prohibitive for all the graph sizes considered in this work except for the graph \texttt{z-alue7065}, which we mainly use to conduct experiments with different parameter choices.

\begin{table*}
\caption{Experiments on different graphs using algorithm~\cite{coudert:hal-03201405} with time limit of 60h (which are 216\,000 seconds) and memory limit of 192GB and the algorithm proposed in this paper {(including the computation of the dominating sets and the hub labeling)}. A skull ($\skull$) in the time/memory column means the time/memory limit was reached before terminating.}
\label{tab:experiments}
\centering
  \resizebox{\linewidth}{!}{
\begin{tabular}{|l|rrr|rrrrr|}
\hline
\textbf{Graph} & \multicolumn{3}{c|}{\textbf{Algorithm~\cite{coudert:hal-03201405}}} & \multicolumn{5}{c|}{\textbf{This paper}} \\ \hline
 & \textbf{time (s)} & \textbf{memory} & \textbf{hyperb.} & \textbf{time (s)} & \textbf{memory} & \textbf{hyperb.} &   \multicolumn{2}{c|}{\textbf{Parameters}} \\
&&&&&&& dom dist & ratio \\
\hline
\hline
\texttt{NotreDame} & 4\,514 & 53.02 GB & 8.0 & 249 & 3.27 GB & 8.0 & 2 & 2 \\
\texttt{web-Stanford} & 8\,249 & 23.28 GB & 23.0 & 15 & 3.22 GB & 23.0 & 8 & 2 \\
\texttt{web-BerkStan} & 65\,134 & 76.93 GB & 23.0 & 59 & 3.33 GB &  23.0 & 4 & 4 \\
\hline
\texttt{t.CAL} & 29\,358 & $\skull$ & [379.0, 1025.5] &  119\,055 & 22.00 GB & 506.5 & 50 & 1.5 \\
\texttt{t.FLA} & $\skull$ & 143.3 GB & [81.0, 818.0] & 1\,199\,907 & 18.04 GB &  229.5 & 25 & 1.5 \\
\texttt{roadNet-PA} & $\skull$ & 148.2 GB & [109.0, 370.5] & 1\,357\,512 & 23.32 GB & 170.5 & 20 & 1.5 \\
\hline
\texttt{buddha} & $\skull$ & 88.35 GB & [93.0, 211.5] & 134\,421 & 52.84 GB & 112.0 & 8 & 1.5 \\
\texttt{froz} & $\skull$ & 106.4 GB & [387.5, 599.0] & 16\,011 & 11.74 GB & 401.5 & 27 & 1.5 \\
\hline
\texttt{grid300-10} & 10 & 1.08 GB & 280.0 & 23 & 5.28 GB & 280.0 & 10 & 1.5 \\
\texttt{grid500-10} & 95 & 2.99 GB & 463.0 & 98 & 6.14 GB & 463.0 & 10 & 2 \\
\texttt{z-alue7065} & 33 & 431.18 MB & 138.0 & 1\,927 & 3.48 GB & 138.0 & 2 & 2 \\
\hline
\end{tabular}
}
\end{table*}

To compare our new approach with~\cite{coudert:hal-03201405}, we let both algorithms run on the same graphs. The parameters of our approach are manually chosen and we give more details about this choice in \S~\ref{sec:exp_paramchoice}. We show the results of our experiments in Table~\ref{tab:experiments}. The most notable points of these results are:
\begin{itemize}
\item For the first three graphs in the experiments (\texttt{NotreDame}, \texttt{web-Stanford}, \texttt{web-BerkStan}), we reduce the memory consumption compared to~\cite{coudert:hal-03201405} by factors of 16.2, 7.2, and 23.1. The running time is reduced by factors of 18.1, 549.9, and 1104, i.e., up to 3 orders of magnitude.
\item Computing the hyperbolicity of the second and third set of graphs (\texttt{t.CAL}, \texttt{t.FLA-w}, \mbox{\texttt{roadNet-PA}}, \texttt{buddha}, \texttt{froz}) was not feasible using previous algorithms but it can be computed using our approach. Especially, we are able to compute the hyperbolicity of a real-world graph with more than a million nodes for the first time. We want to highlight that the memory consumption of our algorithm is below 25 GB for all graphs except \mbox{\texttt{buddha}}. Even in the cases where the algorithm of~\cite{coudert:hal-03201405} hits the time limit of 60 hours and not the memory limit, increasing the time limit most probably does not make these graphs attainable, as the lower and upper bounds are still very far from matching.
\item The class of graphs where our algorithm fails to outperform previous work is grid-like graphs (\texttt{grid300-10}, \texttt{grid500-10}, \texttt{z-alue7065}). This is probably due to the notion of far-apart pairs introduced in~\cite{BorassiCCM15} and also used in~\cite{coudert:hal-03201405}, which we do not use in our approach. Computing the hyperbolicity iterating over far-apart pairs is very fast on grid-like graphs as they contain only very few far-apart pairs --- a perfect grid actually just contains two far-apart pairs, which also form the quadruple that implies its hyperbolicity.
\end{itemize}

\begin{figure}[htbp]
	\centering
	\includegraphics[width=.5\linewidth]{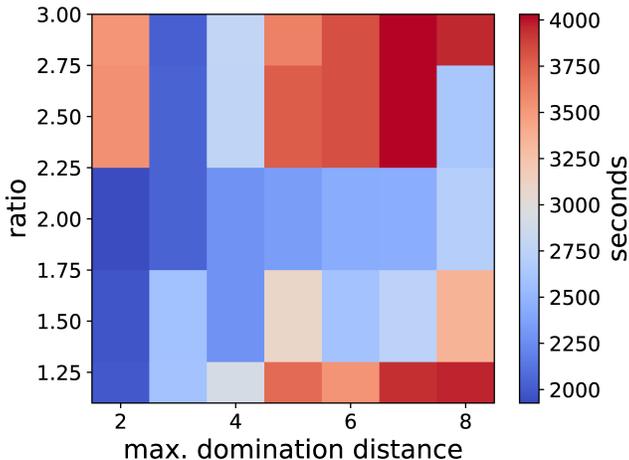}
	\caption{The running times for the \texttt{z-alue7065} graph for different parameters.}
	\label{fig:zalue}
	%\vspace{-5mm}
\end{figure}

\subsection{Parameter Choice.} \label{sec:exp_paramchoice}

The main difficulty for applying our algorithm to new graphs is to determine the parameters that lead to the best performance. Note that these parameters only influence the running time and memory consumption of our algorithm but not the output. As described in \S~\ref{sec:exp_parameters}, we use the two parameters \emph{maximal domination distance} and \emph{ratio}, which in turn determine the dominating set distances that are used. To understand the effect on the running time, we ran our algorithm with different parameters on the graph \texttt{z-alue7065}. More precisely, we did an experiment for all combinations of the maximal domination distance being from the set $\{2, 3, \dots, 8\}$ and the ratio being from the set $\{1.1,\, 1.5,\, 2,\, 2.5,\, 3\}$. The running times of these experiments are shown in Figure~\ref{fig:zalue}. The results suggest that too large or too small values for the ratio are detrimental to a fast running time. Furthermore, on this graph a small maximal domination distance is better. As can be seen by the parameter choices in Table~\ref{tab:experiments}, the moderate value for the ratio is also a good choice for other graphs, while the maximal domination distance can vary greatly depending on the graph. We leave it to future research to investigate how specific graph features influence the best parameter choice.

\begin{figure}[htbp]
\centering
	\begin{subfigure}{.48\linewidth}
\centering
	\includegraphics[width=\linewidth]{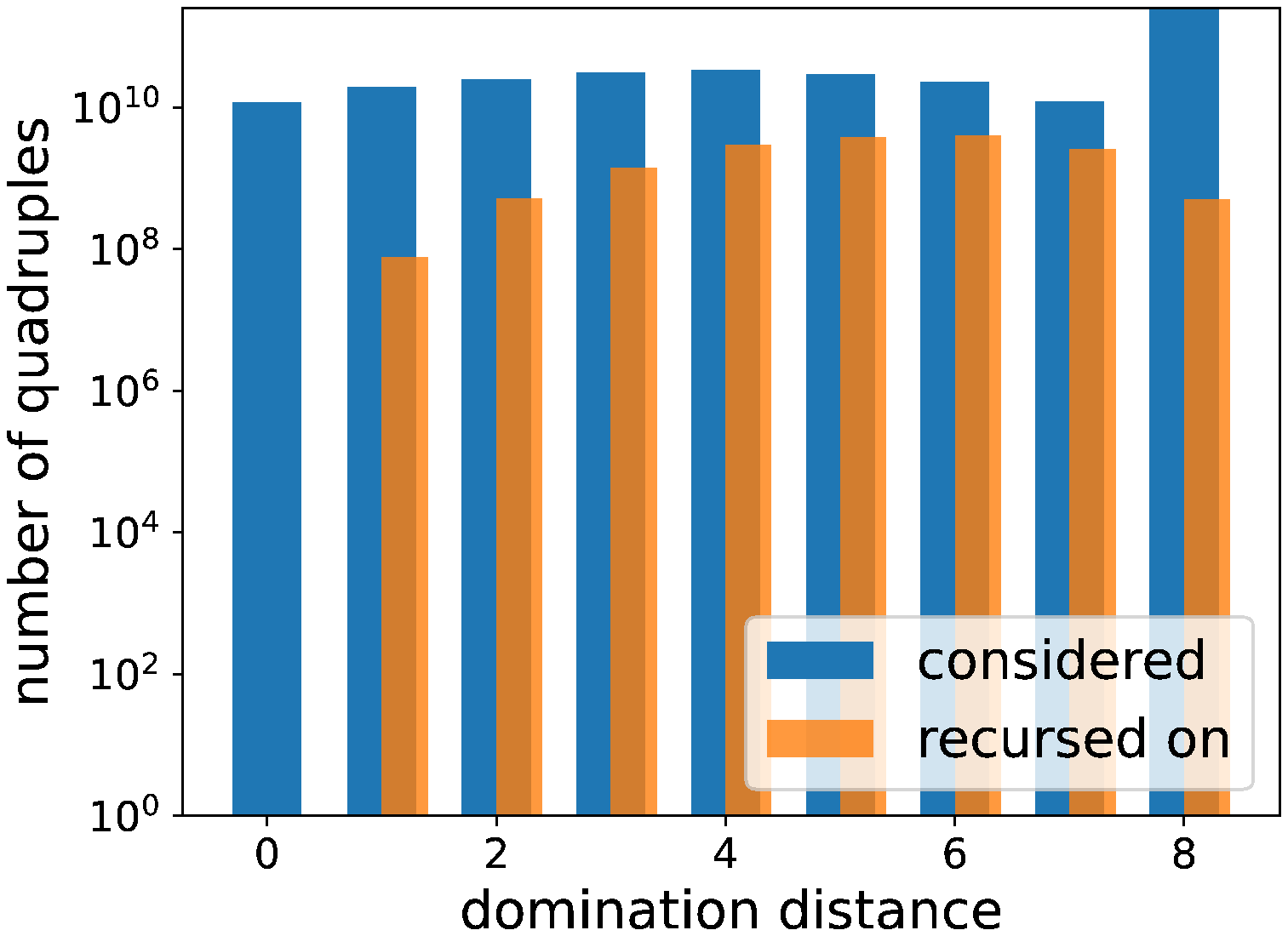}
	\vspace{-3mm}
	\caption{\texttt{buddha:} maximal domination distance $8$ and ratio $1.01$.}
	\label{fig:buddha101}
	\end{subfigure} %\\
	\hfill
	\begin{subfigure}{.48\linewidth}
\centering
	\includegraphics[width=\linewidth]{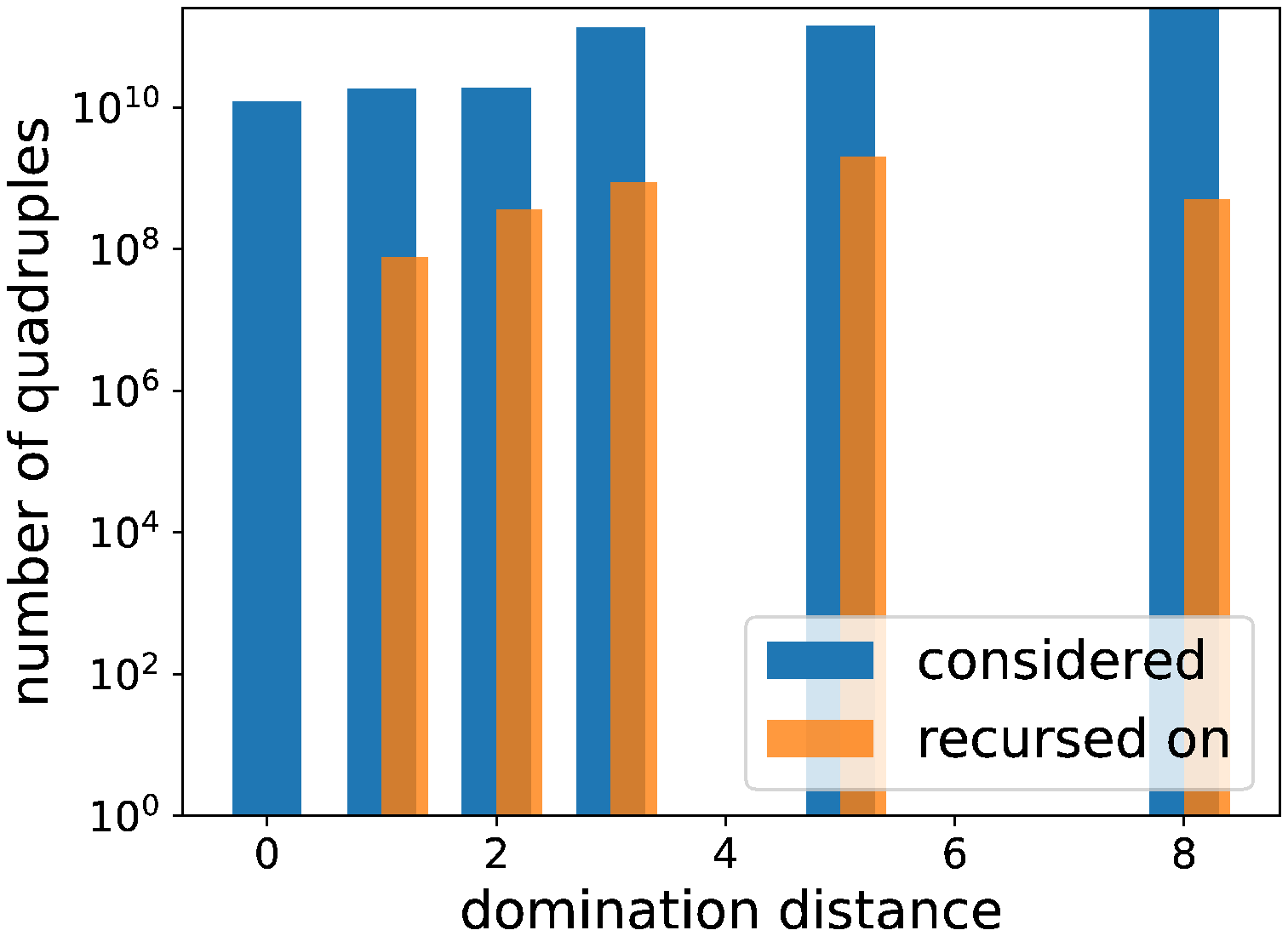}
	\vspace{-3mm}
	\caption{\texttt{buddha:} maximal domination distance $8$ and ratio $1.5$.}
	\label{fig:buddha15}
	\end{subfigure} \\
	\begin{subfigure}{.48\linewidth}
\centering
	\includegraphics[width=\linewidth]{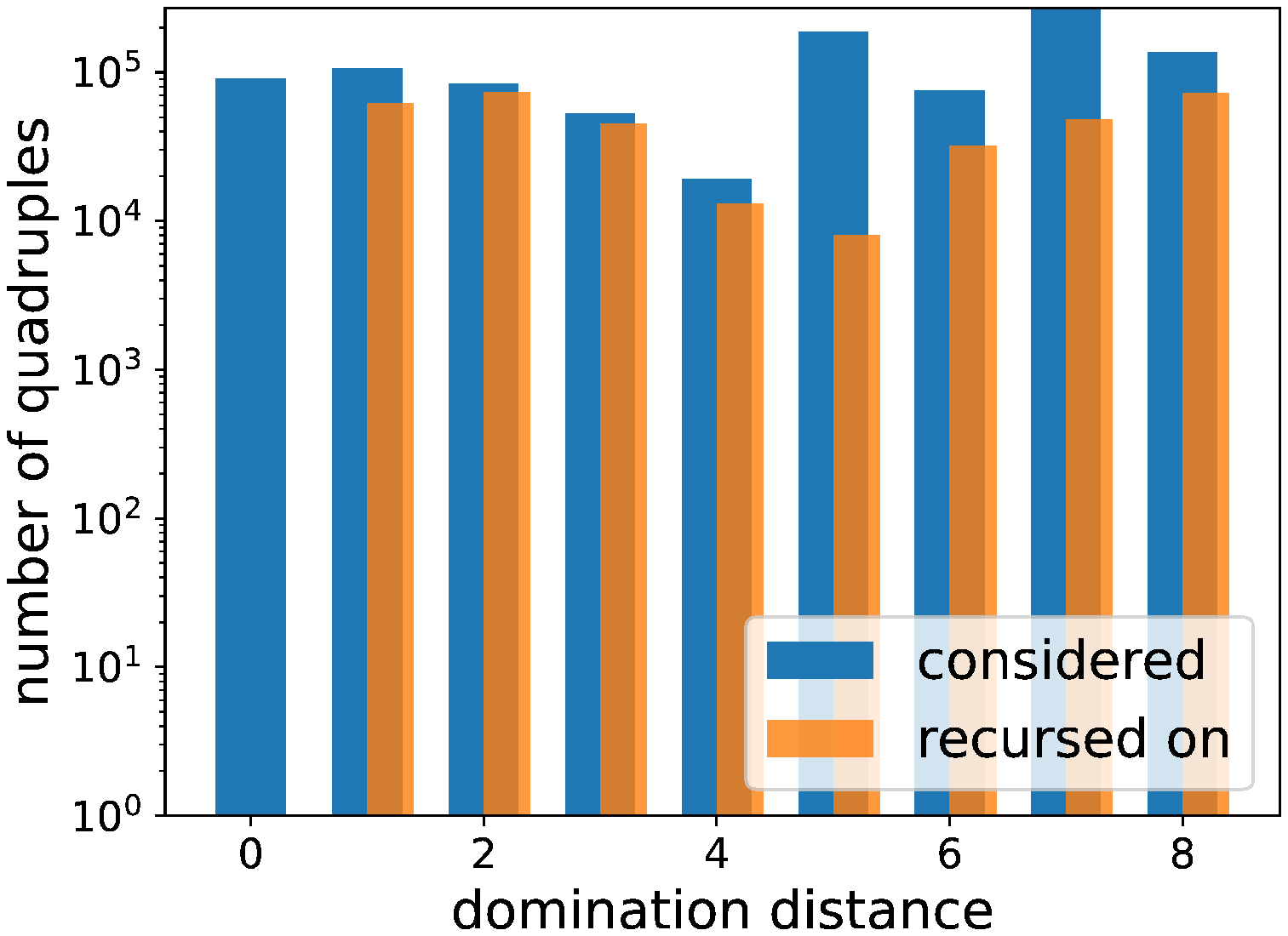}
	\vspace{-3mm}
	\caption{\texttt{web-Stanford:} maximal domination distance $8$ and ratio $1.1$. Note that for this graph the number of considered quadruples and recursive calls are not unimodal.}
	\label{fig:stanford}
	\end{subfigure} %\\
	\hfill
	\begin{subfigure}{.48\linewidth}
\centering
	\includegraphics[width=\linewidth]{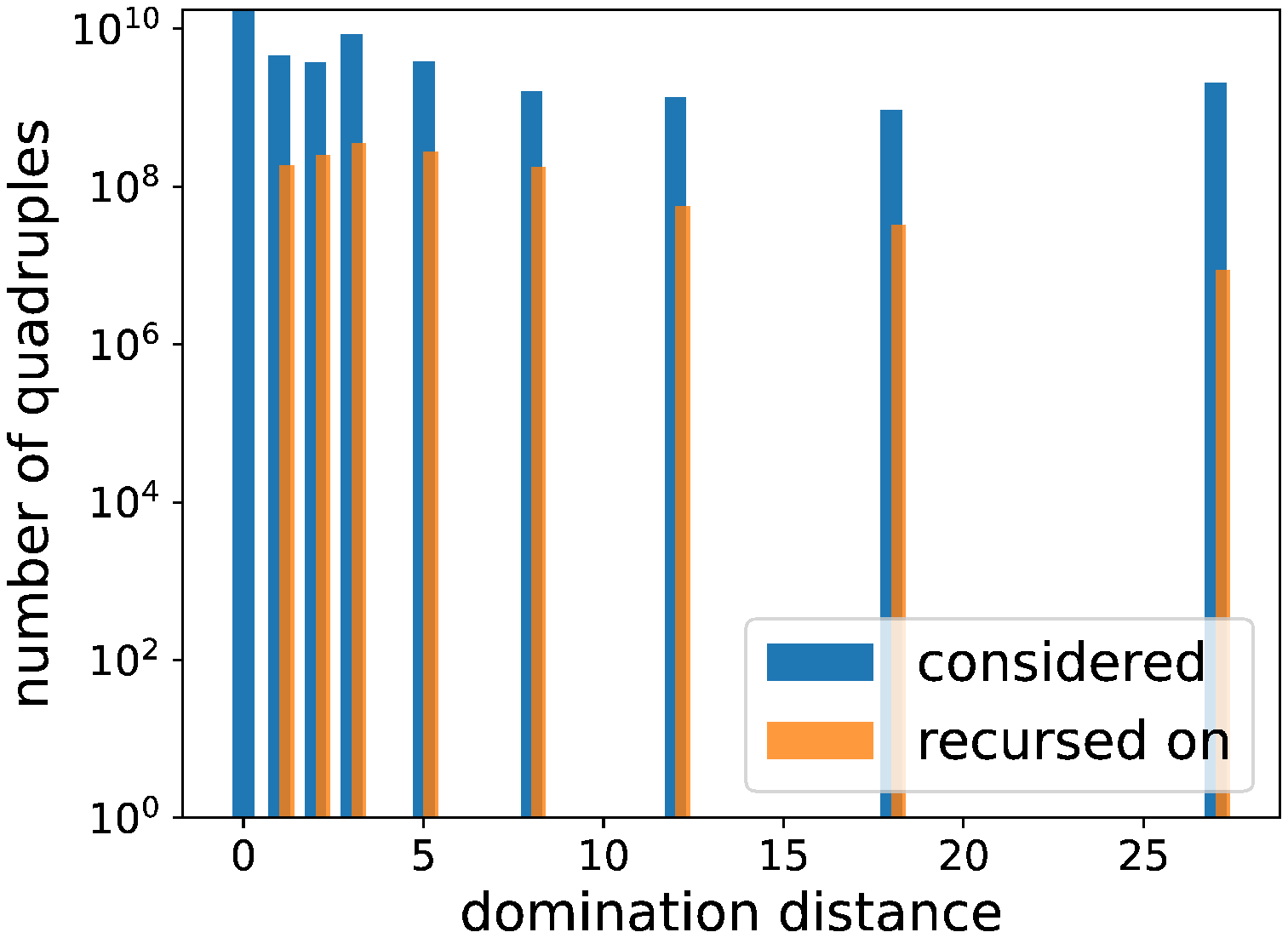}
	\vspace{-3mm}
	\caption{\texttt{froz:} maximal domination distance $27$ and ratio $1.5$. Note that the maximal number of recursive calls is occurring at a quite small domination distance.}
	\label{fig:froz}
	\end{subfigure}
	\caption{The number of considered quadruples and the number of quadruples that we call \texttt{explore} on, for each level of the hierarchy of dominating sets.}
	\label{fig:quad}
\end{figure}

To understand how different sequences of domination distances help to reduce the running time, we conduct an experiment that examines the number of recursive calls on different levels of the hierarchy of dominating sets. To this end, we plot the number of considered quadruples and the number of quadruples that we call \texttt{explore} on (i.e., that we recurse on) for each level of the hierarchy of dominating sets, see Figure~\ref{fig:quad}. Let us first consider the figures~\ref{fig:buddha101} and~\ref{fig:buddha15}, which are both runs on the \texttt{buddha} graph: one with ratio 1.01 and the other with ratio 1.5. Note that the run with the larger ratio skips some domination distances, however, the shape and magnitude of the number of recursive calls does not significantly change. Thus, the pruning of recursive calls works equally well, even when omitting domination distances. This explains how our technique can significantly reduce the computation time. When considering Figure~\ref{fig:stanford}, we can see that the number of recursive calls does not have to be unimodal, but can also exhibit different distributions. Furthermore, as can be seen in Figure~\ref{fig:froz}, the maximal number of recursive calls can also occur at quite small domination distances, even though we need to start with a large one to have a fast overall running time.

% \begin{table}
% \caption{The best runs on different graphs.}
% \label{tab:experiments}
% \centering
% \begin{tabular}{|l|cccc|cc|}
% \hline
% \textbf{Graph} & \textbf{Memory} & \textbf{Time (s)} & \textbf{Old Time (s)} & \textbf{Hyperbolicity} & \multicolumn{2}{c|}{\textbf{Parameters}} \\
% &&&&& dom dist & ratio \\
% \hline
% \hline
% notreDame-d & 3.27 GB & 249.41 & 4514 & 8.0 & 2 & 2 \\
% web-Stanford & 3.22 GB & 15.47 & 8371 & 23.0 & 8 & 2 \\
% y-BerkStan-d & 3.33 GB & 58.71 & 66224 & 23.0 & 4 & 4 \\
% \hline
% t.CAL-w & 22.00 GB & 119055 & --- & 506.5 & 50 & 1.5 \\
% t.FLA-w & 18.04 GB & 1199907 & --- & 229.5 & 25 & 1.5 \\
% roadNet-PA & 23.32 GB & 1357512 & --- & 170.5 & 20 & 1.5 \\
% \hline
% buddha-w & 52.84 GB & 134421 & --- & 112.0 & 8 & 1.5 \\
% froz-w & 11.74 GB & 16011 & --- & 401.5 & 27 & 1.5 \\
% \hline
% grid300-10 & 5.28 GB & 23.02 & 10.19 & 280.0 & 10 & 1.5 \\
% grid500-10 & 6.14 GB & 97.74 & 100.3 & 463.0 & 10 & 2 \\
% z-alue7065 & 3.48 GB & 1927.22 & 39.67 & 138.0 & 2 & 2 \\
% \hline
% \end{tabular}
% \end{table}

%%%%%%%%%%%%%%%%%%%%%%%%%%%%%%%%%%%%%%%%%%%%%%%%%%%%%%
%%%%%%%%%%%%%%%%%%%%%%%%%%%%%%%%%%%%%%%%%%%%%%%%%%%%%%
%%%%%%%%%%%%%%%%%%%%%%%%%%%%%%%%%%%%%%%%%%%%%%%%%%%%%%

\section{Conclusion}\label{sec:conclusion}

In this work, we presented a new practical algorithm to compute the graph hyperbolicity of very large graphs (compared to the graph sizes that were feasible for previous algorithms). The main idea of our algorithm is to construct a hierarchy of dominating sets and use these to prune a large number of quadruples that need to be considered as candidates for the hyperbolicity of the graph. Our approach is especially suited for graphs with small dominating sets and a hyperbolicity that is not very small. For the first time we can compute the hyperbolicity of non-trivial graphs with more than a million nodes and we obtain speed-ups of three orders of magnitude while reducing the memory consumption by more than an order of magnitude.

A problem that we leave open is how to derive the parameters that lead to the lowest running time and memory consumption for our algorithm in an automated fashion. 
Another open question is how to use the notion of far-apart pairs to further reduce the number quadruples to consider.

%%%%%%%%%%%%%%%%%%%%%%%%%%%%%%%%%%%%%%%%%%%%%%%%%%%%%%
%%%%%%%%%%%%%%%%%%%%%%%%%%%%%%%%%%%%%%%%%%%%%%%%%%%%%%
%%%%%%%%%%%%%%%%%%%%%%%%%%%%%%%%%%%%%%%%%%%%%%%%%%%%%%

% remove if not needed anymore
% \clearpage

\bibliographystyle{plainurl}% the mandatory bibstyle
\bibliography{biblio}

\end{document}